\documentclass[11pt]{article}
\usepackage[letterpaper,margin=1in]{geometry}
\pdfoutput=1

\usepackage{amsfonts} 
\usepackage{amsthm, amssymb}

\usepackage{graphicx}
\usepackage{adjustbox}       
\usepackage{multicol}            
\usepackage{url}
\usepackage{enumerate, enumitem}  
\usepackage{subcaption}
\usepackage[usenames]{xcolor}

\usepackage{hyperref}
\hypersetup{
    colorlinks=true,
    linkcolor=red,
    citecolor=blue,
    urlcolor=violet,
    pdftitle={Optimal Sketching Bounds for Sparse Linear Regression},
    linktocpage=true,
}
\usepackage{amsmath, nicefrac}
\usepackage[noabbrev,capitalize]{cleveref}
\usepackage{thm-restate}

\usepackage{mathtools}
\usepackage{multirow}
\usepackage{algorithmic}
\usepackage{algorithm}

\usepackage[style=alphabetic,natbib=true]{biblatex}
\DeclareUnicodeCharacter{0301}{\'{e}}

\colorlet{darkgreen}{green!45!black}

\newcommand{\R}{\mathbb{R}}

\newcommand{\eps}{\varepsilon}
\renewcommand{\epsilon}{\varepsilon}

\newcommand{\calN}{\mathcal{N}}

\newcommand{\norm}[1]{\|#1\|}

\newcommand{\normm}[2]{\lVert#2\rVert_{#1}}

\newcommand{\E}{\mathbb{E}}
\newcommand{\pr}{\mathbb{P}}

\newcommand{\poly}{\operatorname{poly}}

\newcommand{\polylog}{\mathrm{polylog}}

\DeclareMathOperator{\argmin}{argmin}
\DeclareMathOperator{\relu}{ReLU}
\DeclareMathOperator{\supp}{supp}
\newcommand{\eqdef}{\mathbin{\stackrel{\rm def}{=}}}

\newtheorem{lemma}{Lemma}

\newtheorem{theorem}{Theorem} 

\newtheorem{corollary}[theorem]{Corollary}

\newtheorem{definition}{Definition}

\newcounter{sideremark}

\title{Optimal Sketching Bounds for Sparse Linear Regression}
\author{Tung Mai\thanks{Adobe Research, \href{mailto:tumai@adobe.com}{tumai@adobe.com}}
\and Alexander Munteanu\thanks{TU Dortmund University, \href{mailto:alexander.munteanu@tu-dortmund.de}{alexander.munteanu@tu-dortmund.de}}
\and Cameron Musco\thanks{University of Massachusetts Amherst, \href{mailto:cmusco@cs.umass.edu}{cmusco@cs.umass.edu}}
\and Anup B. Rao\thanks{Adobe Research, \href{mailto:anuprao@adobe.com}{anuprao@adobe.com}}
\and Chris Schwiegelshohn\thanks{Aarhus University, \href{mailto:cschwiegelshohn@gmail.com}{cschwiegelshohn@gmail.com}}
\and David P. Woodruff\thanks{Carnegie Mellon University, \href{mailto:dwoodruf@cs.cmu.edu}{dwoodruf@cs.cmu.edu}}
}
\date{\today}

\include{bib_macros}
\begin{document}
\allowdisplaybreaks
\maketitle
\begin{abstract}
	We study oblivious sketching for $k$-sparse linear regression under various loss functions. In particular, we are interested in a distribution over sketching matrices $S\in\R^{m\times n}$ that does not depend on the inputs $A\in\R^{n\times d}$ and $b\in\R^n$, such that, given access to $SA$ and $Sb$, we can recover a $k$-sparse $\tilde x\in\mathbb{R}^d$ with $\|A\tilde x-b\|_f\leq (1+\varepsilon) \min\nolimits_{k{\text{-sparse}\,x\in\mathbb{R}^d}} \|Ax-b\|_f$. Here $\|\cdot\|_f: \mathbb R^n \rightarrow \mathbb R$ is some loss function -- such as an $\ell_p$ norm, or from a broad class of hinge-like loss functions, which includes the logistic and ReLU losses.	
 
	We show that for sparse $\ell_2$ norm regression, there is a distribution over oblivious sketches with $m=\Theta(k\log(d/k)/\varepsilon^2)$ rows, which is tight up to a constant factor. This extends to $\ell_p$ loss with an additional additive $O(k\log(k/\varepsilon)/\varepsilon^2)$ term in the upper bound.  This establishes a surprising separation from the related sparse recovery problem, which is an important special case of sparse regression, where $A$ is the identity matrix.
	For this problem, under the $\ell_2$ norm, we observe an upper bound of $m=O(k \log (d)/\varepsilon + k\log(k/\varepsilon)/\varepsilon^2)$, showing that sparse recovery is strictly easier to sketch than sparse regression. 
 
	For sparse regression under hinge-like loss functions including sparse logistic and sparse ReLU regression, we give the first known sketching bounds that achieve $m = o(d)$ showing that $m=O(\mu^2 k\log(\mu n d/\varepsilon)/\varepsilon^2)$ rows suffice, where $\mu$ is a natural complexity parameter needed to obtain relative error bounds for these loss functions. We again show that this dimension is tight, up to lower order terms and the dependence on $\mu$.	
 
	Finally, we show that similar sketching bounds can be achieved for LASSO regression, a popular convex relaxation of sparse regression, where one aims to minimize $\|Ax-b\|_2^2+\lambda\|x\|_1$ over $x\in\mathbb{R}^d$. We show that sketching dimension $m =O(\log(d)/(\lambda \varepsilon)^2)$ suffices and that the dependence on $d$ and $\lambda$ is tight.
\end{abstract}
\thispagestyle{empty}

\clearpage
\tableofcontents
\thispagestyle{empty}
\clearpage

\setcounter{page}{1}
\section{Introduction}\label{sec:intro}
We study oblivious sketching for $k$-sparse regression. Given a data matrix $A\in \R^{n\times d}$, and a target vector $b\in \R^n$, linear regression problems aim at finding a vector $x\in \R^{d}$ such that $Ax\approx b$.
The deviation from $Ax$ to $b$ is quantified via a loss function $f(Ax,b)$, where popular examples include the loss in terms of an $\ell_p$ norm, logistic loss, and {ReLU}.
Sketching techniques for these problems have been widely and successfully applied. Here, one samples a sketching matrix $S\in \R^{m\times n}$ from some distribution and attempts to solve the problem on $SA$ and $Sb$. If $S$ can be sampled from a distribution that does not depend on either $A$ or $b$, we call the sketch \emph{oblivious}.
Our aim is to minimize the target dimension $m$, called the \emph{sketching complexity} or \emph{sketching dimension}, while retaining the ability to extract a $(1+\varepsilon)$-approximate solution $\tilde x$ by only using the smaller sketch rather than the original data.
A particularly desirable property for the model parameter $x$ is \emph{sparsity}, i.e., where $x$ is restricted to have at most $k$ non-zero elements. 
Sparse linear regression is an important technique for handling very high dimensional data sets, such as those arising in biostatistics \citep{Mbatchou2021}. It produces a linear model that depends on just a small number of parameters, and thus is more interpretable and can be learned accurately from a relatively small data sample. From a computational point of view, imposing sparsity constraints makes the problem significantly harder. Most unconstrained regression problems are convex and thus we can draw on a wide array of gradient-based methods, but sparse linear regression is $NP$-hard \citep{Natarajan95}. Many heuristics and relaxations have been developed to solve the problem in practice,  \citep[see][]{Tibshirani96,miller2002subset,DasK08,DasK18}. Their running time typically depends at least polynomially on the number $n$ of observations. It is thus natural to study methods for reducing the number of rows in the data so that computations become more efficient and an approximately optimal solution is retained.

While computational savings are immediate from our results for most of the mentioned heuristics and approximation algorithms covered in the related work (see \cref{sec:related_work}), potential speed-ups need to be verified on a case-wise basis and contrasted with the cost of applying the sketch to the data. Additionally, sketching has other motivations such as saving memory, processing (turnstile) data streams \citep{ClarkW09}, and aggregating distributed data \citep{KannanVW14}. We refer to \citep{Munteanu23} for a brief introduction. We stress that our bounds refer only to the reduced sketching dimension $m$, not to other complexity measures.

While the performance of sketching techniques is well understood for \emph{unconstrained} regression problems, we know little about the complexity for sparse regression problems. It is clear that a guarantee for unconstrained regression also applies to sparse regression, but it is not at all clear that these bounds are tight. In particular, the special case of sparse recovery has many celebrated results exploiting sparsity to reconstruct a target signal from a few measurements. Thus we ask the following question:

\begin{center}
	\emph{Can oblivious sketching techniques benefit from\\model-sparsity for various regression problems?}
\end{center}

To our knowledge, the above question has not been investigated previously. In this paper we answer this question in the affirmative for a large class of loss functions, including the $\ell_p$ loss $\|Ax-b\|_p= \sqrt[p]{\sum_{i=1}^n |A_ix-b_i|^p}$ for $p\in[1,2]$, logistic loss $\sum_{i=1}^n \ln(1+\exp(-b_i\cdot A_ix))$, ReLU $\sum_{i=1}^n \max(0,-b_i\cdot A_ix)$ and further hinge-like loss functions $f$, where we define $\|Ax-b\|_f= \sum_{i=1}^n f(A_ix-b_i)$.
We also investigate the sketching complexity for popular heuristics used to impose sparsity, such as LASSO regression \citep{Tibshirani96}, where instead of strictly forcing $x$ to be $k$-sparse, we use the $\ell_1$ penalized objective function $\min_{x\in \R^d} \|Ax-b\|_2^2 + \lambda \|x\|_1$ to find a sparse solution. For all of these problems, we obtain optimal or nearly optimal sketching bounds. Crucially, the dependence on $d$ is at least linear for unconstrained regression, but it appears only logarithmically for sparse regression.
Thus our paper makes significant steps towards exploiting the power of sparsity for sketching regression problems.

Let $\Psi_k = \{x \in \R^d \mid \|x\|_0 \leq k\}$ denote the set of $k$-sparse vectors in $\R^d$.
We also study the related problem of $k$-sparse affine embeddings, where given the sketch and an estimator $E(SA,Sb,x)$ we require that

$$\forall x\in \Psi_k: (1-\eps) \|Ax-b\|_f \leq E(SA,Sb,x) \leq (1+\eps) \|Ax-b\|_f. $$

A $k$-sparse affine embedding is stronger than a sketch for sparse linear regression since it implies the latter. It is a generalization of the restricted isometry property (RIP) studied in the context of compressed sensing. The RIP is a special case where $A=I$ is the identity matrix, $b=0$, and we seek to preserve the norm of any $k$-sparse $x \in \Psi_k$ up to $1\pm\eps$ distortion.

The problem of $\ell_2$ sparse recovery can also be seen as a special case of sparse $\ell_2$ regression, where $A=I$, and we are given access only to a sketch $Sb$ of the vector $b$. The goal is to recover a $k$-sparse vector $\tilde x$ that is within $(1+\varepsilon)$ error to the minimizing $x$, i.e., it satisfies $\|\tilde x - b\|_2 \leq (1+\varepsilon) \min\nolimits_{x\in \Psi_k}\|x - b\|_2.$
Sketching sparse regression seems very similar to sparse recovery, since similar methods are available that yield similar upper bounds for both problems. However, our studies imply a surprising separation result. Namely, the sparse regression problem is \emph{strictly harder} to sketch than sparse recovery.

One might wonder why we do not consider data dependent sampling algorithms in addition to oblivious linear sketches, since sampling techniques are important tools for approximating regression problems in the non-sparse setting. This is because it can be observed (see \cref{thm_inf:samplingLower,thm_inf:samplingLower2}) that sampling does not help in the case of sparse regression; sampling roughly the entire input is necessary to achieve any non-trivial bound.

The crucial advantage of sketching over sampling seems to be the property of \emph{obliviousness} to the subspaces that need to be embedded, which allows us to take a union bound over all $k$-subsets of coordinates. Sampling algorithms, however, would need a different measure for each possibility. While data dependent importance sampling techniques are widely successful for the unconstrained non-sparse regression problems, they do not give any non-trivial bounds in the sparse setting. This underlines the importance of oblivious sketching techniques in the sparse context.

\subsection{Related work}\label{sec:related_work}
An important special case of the $k$-sparse regression problem is compressed sensing \citep{CandesRT06,Donoho06,BaraniukDDW08} where the matrix $A$ is the $n\times n$ identity matrix and one seeks to find a sparse vector $x$ that represents the non-sparse signal $b$ well using a linear sketching (or sensing) matrix. It is well known that a matrix for that problem requires $m=\Omega(k\log(d/k))$ rows \citep{BaIPW10}, which was improved to $m=\Omega(k\log(d/k)/\eps+k/\eps^2)$ by \citet{price20111+}. As for upper bounds, it is known that a (Gaussian) RIP matrix can be constructed with $m=O(k\log(d/k))$ rows, which suffices to solve the problem for constant $\epsilon$. Recently, new and tighter proofs for the Gaussian construction appeared \citep{LiXG20} which implicitly yield $m=\Theta(k\log(d/k)/\eps^2)$. However, to solve sparse recovery, neither the Gaussian matrix nor the RIP are necessary. For instance, when sparsity $2k$ is allowed, the problem can be solved, with constant probability for any particular input, using CountSketch \citep{CharikarCF04} with  $m=O(k\log(d/k)/\eps)$ rows \citep{price20111+}.

To our knowledge the generalization of oblivious linear sketching for sparse linear regression has not been investigated before. However, there is a body of work on the column selection problem for sparse linear $\ell_2$ regression. The sparse regression problem, i.e., minimizing the regression cost over all $k$-sparse vectors, is $NP$-hard \citep{Natarajan95}, and under reasonable complexity-theoretic assumptions it is even hard to approximate within a significantly stronger bound than the trivial $\|b\|_2$ in quasipolynomial time \citep{FosterKT15}. In light of this, some authors have identified and characterized instances for which widely applied heuristics have performance guarantees \citep{DasK08,DasK18,GilbertMS03,Tropp06b,Tropp06,TroppGS06}.
Another direction that is closer to the used heuristics, is the online version of the column selection problem. It was shown by \cite{FosterKT15} that an online algorithm whose iterations run in polynomial time would imply $NP\subseteq BPP$, even if it is allowed to increase the number of columns by an $O(\log d)$ factor. In light of these impossibility results, research has focused on \emph{inefficient} algorithms. The online algorithm of \cite{FosterKT15} runs in roughly $O(k\binom{d}{k})$ time per iteration. This was improved by \citet{HarPeledIM18} to roughly $\tilde{O}(d^{k-1})$ by giving a data structure that approximates, to within a $(1+\eps)$ factor, the geometric distance query to the closest $(k-1)$-dimensional flat spanned by the input points, leveraging the geometric interpretation of sparse linear regression. The same reference gives an impossibility result of roughly $\tilde\Omega(d^{k/2}/e^k)$ for any multiplicative error approximation by reducing from the $k$-SUM problem. Assuming the RIP property, \cite{KaleKLP17} give an efficient online algorithm with guarantees.

The heuristic that is arguably most used in practice for solving sparse regression is the least absolute shrinkage and selection operator (LASSO) by \citet{Tibshirani96}. It introduces a convex relaxation of the $\ell_0$ constraint, replacing it by an $\ell_1$ constraint. 
It was shown for well-behaved matrices $A$, that the LASSO algorithm recovers a (nearly optimal) sparse solution \citep{CandesRT06,Donoho06}.
The resulting Lagrangian form is $\min_{x\in \R^d} \|Ax-b\|_2^2 + \lambda \|x\|_1$. Since the $\lambda \|x\|_1$ regularization term is non-negative, it is immediate that an $\ell_2$-subspace embedding is sufficient for preserving the cost of any $x\in\R^d$ within $(1\pm\eps)$ multiplicative error with $m=\Theta(d/\eps^2)$ rows \citep{Sarlo2006,NelsonN14}. To our knowledge there are no results on sketching this objective with fewer rows by exploiting the sparsity induced by the regularizer.

For dense linear $\ell_p$ regression there are numerous sketching results. Starting with $\ell_2$, \cite{Sarlo2006} showed that $m=O(d/\eps)$ rows suffice to preserve the minimizer up to $1+\eps$ error. This was complemented by a matching lower bound by \cite{ClarkW09}. Extensions to $\ell_p$ regression via oblivious $\ell_p$ subspace embeddings and sampling were given in \citep{Clarkson05subgradient,DasguptaDHKM09,SohlerW11,WoodruffZ13,ClarksonW15Mestim,WangW19,LiW021}. Recent works \citep{MunteanuOW21,MunteanuOW23} gave the first oblivious linear sketches for logistic regression and (implicitly) for the $\relu$ function. Importance sampling algorithms for those generalized linear regression problems were developed by \cite{MunteanuSSW18} and further improved and generalized \citep{mai2021coresets,MunteanuOP22,WoodruffY23}.

\section{Our techniques and results}\label{sec:results}
Our results are summarized in \cref{tab:results} and cover bounds on the reduced sketching dimension $m$, not on other complexity measures such as computational cost. Our upper bounds are for affine embeddings, so an algorithm using our sketch enjoys approximation bounds over the entire search space. This straightforwardly implies matching minimization upper bounds. Our lower bounds are for the more challenging minimization variants, except for the hinge-like losses. Our bounds are tight for sparse $\ell_2$ regression. The generalization to $\ell_p$ is tight up to an \emph{additive} $O(k\log(k/\eps)/\eps^2)$ term; specifically, this means our result is tight for reasonably small $k = O(\sqrt{\eps d})$. The \rm{ReLU} upper bound has another $\tilde{O}(\mu^2)$ factor, in addition to the gap reported for $\ell_1$. Here $\mu$ is a natural parameter that is needed to parameterize the complexity of compressing data below $\Omega(n)$ for those losses \citep{MunteanuSSW18}. This also holds for hinge-like loss functions (including logistic loss) which adds another additive $\log n$, and the complementing lower bound is slightly weaker in the sense that it holds for $k$-sparse affine embeddings instead of the minimization problem.

Any subsampling approach has a matching $\Theta(n)$ bound for any loss function.  
The lower bound is given in the minimization setting and for subspace embeddings with different levels of obliviousness to the data. 
For LASSO with regularization parameter $\lambda$, our upper bound is $m=O(\log(d)/(\lambda\eps)^2)$, which we complement by a lower bound of $m=\Omega(\log(\lambda d)/\lambda^2)$. 
Finally, our upper bound for $\ell_2$ sparse recovery leaves only a small additive gap of $O(k\log(k/\eps))$ to the best known lower bound. More interestingly, the bound is sufficient to yield a separation from the \emph{strictly harder} sparse regression problem.
In summary, our bounds are tight up to lower order (additive) terms with general parameterizations, and they are tight for reasonably small values of $k/\eps$.

\begin{table*}[t!]
\caption{Summary of results. Here 
		$\mu$ is a data dependent complexity parameter \citep{MunteanuSSW18}, and $\lambda$ is the regularization parameter for LASSO, see the main text for details.} 
	\label{tab:results}
	\centering
	\begin{tabular}{ | l| l| l|}
		\hline
		{\bf Loss function / type} & {\bf Bound} & {\bf Reference} \\ \hline \hline
		${\ell_2}$ & $\Theta(k\log(d/k)/\eps^2)$ & \cref{thm_inf:LB_l2,thm_inf:UB_l2} \\ \hline\hline
		
		$\ell_p, p\in[1,2)$ & ~ & ~\\ \hline
		Lower bound & $\Omega(k\log(d/k)/\eps^2)$ & \cref{thm_inf:LB_lp} \\ \hline
		Upper bound & $O(k\log(d/k)/\eps^2+ k\log(k/\eps)/\eps^2)$ & \cref{thm_inf:UB_l1_lp} \\ \hline \hline
		
		$\relu$ \& hinge-like
		& ~ & ~ \\ \hline 
		Lower bound & $\Omega(k\log(d/k)/\eps^2)$ & \cref{thm_inf:LB_relu,thm_inf:LB_hingelike} \\ \hline
		Upper bound for $\relu$ & $O({\mu^2 k} \log ({\mu d}/{\eps})/{\eps^2})$ & \cref{thm_inf:UB_relu_from_l1} \\ \hline 
		Upper bound for hinge-like & $O({\mu^2 k} \log ({\mu n d}/{\eps})/{\eps^2})$ & \cref{thm_inf:UB_hingelike_from_relu} \\
		\hline \hline
		Any loss via sampling & $\Theta(n)$ & \cref{thm_inf:samplingLower,thm_inf:samplingLower2}  \\ \hline \hline
		LASSO with $\lambda \norm{x}_1$ & ~ & ~ \\\hline 
		Lower bound & $\Omega(\log(\lambda d)/\lambda^2)$ & \cref{thm_inf:LASSOLower} \\ \hline
		Upper bound & $O(\log(d)/(\lambda\eps)^2)$ & \cref{thm_inf:LASSOUpper} \\ \hline \hline
		$\ell_2$ sparse recovery & ~ & ~ \\ \hline
		Lower bound & $\Omega(k\log(d/k)/\eps+k/\eps^2)$ & {\citep{price20111+}}
		\\ \hline
		Upper bound & $O(k \log(d)/\eps + k\log(k/\eps)/\eps^2)$ & \cref{thm_inf:UB_sparse_recovery} \\ \hline
	\end{tabular}
\end{table*}

In the remainder we present our main results and the main ideas and technical challenges behind their proofs. The formal proofs and details are moved to their respective appendices for a concise presentation.

\subsection{Lower bounds} We obtain our lower bounds by giving a sequence of reductions. Our main lower bound for sparse $\ell_2$ regression (from which the further bounds will be derived) is obtained by a reduction from approximate (constant fraction) support recovery for sparse PCA (sPCA).

We note that a related reduction was given in \citep{BreslerPP18}, whereas the hardness of the approximate sPCA was covered in \citep{CaiMaWu2013}. However, the combination of these prior works does not give the desired hardness result for our problem for the following reasons. Their reduction requires $d$ exact sparse regression solves, where \emph{each} column is regressed on all remaining columns. The decision if the column in iteration $i$ is included in the sparse support is done by comparing the projected norm onto the optimal sparse subspace to a certain fixed threshold. The main issue with this is that they get only a $k^2$ dependence in the reduction from sPCA to sparse regression rather than $k$, which is necessary in their analysis to separate between columns being in the support or not. It is unclear how to replace those steps with randomized decisions without inflating the dependence on $k$ or other parameters even more. Further, introducing randomization in each iteration would yield only a lower bound against $1/d$ error probability as opposed to our lower bound against constant probability.

To prove our optimal lower bound against constant error, we give a reduction by solving only one single regression problem. The arguments and support set constructions from prior work \citep{AminiW2009,CaiMaWu2013} are not directly applicable for the following reasons: we plant the information on the unknown support onto an additional column which we regress onto the standard columns. This additional information is weighted sufficiently high such that it allows us to recover (a constant fraction of) the support.

Crucially, the weight is also sufficiently low, such that the support recovery problem remains hard to solve. But this needs to be reproven using our techniques, which then implies the hardness of the sparse regression problem. We note that full recovery has an $\Omega(k\log(d-k))$ lower bound \citep{AminiW2009}, which is larger than our upper bound on the sparse regression problem. We thus rely on a relaxation to approximate constant fraction support recovery using error correcting codes, which requires us to prove a novel lower bound for the \emph{simplified} problem, where the additional planted information is given to the algorithm.

We build on \citep{AminiW2009} as our starting point. The authors studied sPCA for a \emph{spiked covariance model}, where we take measurements from a Gaussian with a covariance matrix $(I_d + vv^T)$ and $v$ is a $k$-sparse vector. Here, to find the largest eigenvector means that based on vectors drawn from the Gaussian distribution, we need to find the vector $v$. Since all non-zero entries have the same value, this reduces to finding the $k$-sparse support of $v$. The authors show that based on a small number of measurements, this problem is impossible to solve with good probability by information-theoretic arguments. Here we adapt their high level intuition and prove novel hardness bounds for our adapted variant of sPCA.

\paragraph{Sparse $\ell_2$ regression}
We sketch the proof of our main theorem. Several technical details are omitted for brevity of presentation. The detailed technical derivations are in the appendix.

\begin{restatable}{theorem}{thmLBltwo}
\label{thm_inf:LB_l2}
	Let $A\in \R^{n\times d}, b\in \R^n$. Suppose $S\in\R^{m\times n}$ is an oblivious linear sketch for $k$-sparse $\ell_2$ regression with an arbitrary estimator $E(SA,Sb,x)$, such that $\tilde x \in \argmin_{x\in\Psi_k} E(SA,Sb,x)$ satisfies $\|A\tilde x-b\|_2 \leq (1+\eps) \min_{x\in \Psi_k} \|Ax-b\|_2$ with constant probability. Then $m=\Omega(k\log(d/k)/\eps^2)$.
\end{restatable}
Our proof is structured as follows: First, we construct a suitable (hard) distribution over $k$-sparse supports, which is used to define our input distribution. Second, we prove the impossibility of recovering a constant fraction of the support with a small number of measurements (rows) from the input distribution below an information-theoretic lower bound. Third, we construct an $\ell_2$-regression instance for which any $1+\Theta(\eps)$ approximation derived from an oblivious sketch, paired with an arbitrary estimator, reveals a constant fraction of the support. The hardness result thus turns over to the regression problem.

For constructing the hard distribution (first proof step), we construct an error correcting code $\mathcal C$ of roughly size $| \mathcal C | = (\frac d k )^k$ consisting of $k$-sparse binary vectors that overlap in at most $ck$ indexes. The code is also exactly balanced in the following sense. Every single index $i$ appears in exactly the same number of codewords as any other index $i'$, i.e., roughly $(\frac{d}{k})^{k-1}= \frac{k}{d} |\mathcal C|$ times, and each pair $(i,j)$, for $i\neq j$, appears exactly the same number of times as any other pair $(i',j')$, for $i'\neq j'$, i.e., roughly $(\frac{d}{k})^{k-2}= \frac{k^2}{d^2} |\mathcal C|$ times. We augment each codeword by another (w.l.o.g. first) coordinate, which is fixed to $1$.

We pick a codeword $c\in\mathcal C$ uniformly at random and we let our distribution over $n\times d$ inputs be $Z=[b,A]=G(I_d + vv^T)^{1/2}$, where $G$ is a random Gaussian matrix, i.e., each $G_{ij}\overset{i.i.d.}{\sim} N(0,1)$, and for all $j\in[d]\setminus\{1\}$ we set $v_j=\frac{\eps}{\sqrt{k}}c_j$ whereas for the first, augmented coordinate, we set $v_1=c_1=1$.
This concludes the description of the hard input distribution and the situation before sketching.

Now, since we prove a lower bound against arbitrary estimators, which could for example change the basis and rescale in an arbitrary but appropriate way, we can assume w.l.o.g. that the sketching matrix $S\in\mathbb{R}^{m\times n}$ has orthonormal rows. The sketch thus takes the form $$SZ=SG(I_d + vv^T)^{1/2}=H(I_d + vv^T)^{1/2},$$ where $H$ is again a Gaussian matrix. $SZ$ thus has the same Gaussian distribution as the input matrix $Z$ but the number of rows is reduced from $n$ to $m < n$.

We move to the second proof step, i.e., the impossibility of recovering a constant fraction of the support with a small number of measurements (rows) from the above input distribution. To this end, we let $X$ be a random variable that has the distribution of one row of our sketch. 
We use Fano's inequality to bound the failure probability in terms of the size of the code $\mathcal C$ and the mutual information that quantifies how much information the rows of the sketch, denoted by $X^m$, reveal about the unknown support $U$ of $c\in\mathcal C$: 
$$\pr[\texttt{error}] \geq 1 - \frac{I(U;X^m)+\log 2}{\log( |\mathcal C| -1)}.$$

In order to bound the mutual information, we observe by the chain rule for entropy and the maximum entropy property of the Gaussian distribution \citep{CoverT2006} that the mutual information can be bounded in terms of log determinants: 
\begin{align*}
    I(U,X^m) &= H(X^m) - H(X^m|U) \\
    &\leq \frac{m}{2}\log\det \E [xx^T] - \frac{m}{2}\log\det \E [xx^T|U]
\end{align*}
Leveraging the balanced structure of our code construction, the matrices involved have a nice block structure which we exploit to bound the mutual information by $O(\eps^2 m)$. Further note that the logarithmic code size satisfies $\log(|\mathcal C|)=O(k\log(d/k))$. Plugging this back into Fano's inequality we obtain a lower bound $m=\Omega(k\log(d/k)/\eps^2)$ against any constant error probability support recovery algorithm, which concludes the second step of our proof. We refer to the appendix for formal details.

Finally, we describe the reduction to the sparse $\ell_2$ regression problem, i.e., the third and last step of our proof. Given $Z=[b,A]$ as described above, we consider the following instance:
\[
\min_{x \in \Psi_{k}}
\left\lVert
\left[
\begin{array}{c}
	M \, M \, \ldots \, M \\
	A
\end{array}
\right] x 
- 
\left[
\begin{array}{c}
	\sqrt{k} M \\
	b
\end{array}
\right]
\right\rVert_2,
\]
where $M$ is sufficiently large such as to enforce $\sum_i x_i$ to be close to $\sqrt{k}$. In particular this is needed to prevent the trivial solution $x=0$, and more precisely to impose $\|x\|_2\geq \frac{1}{\sqrt{k}}\|x\|_1\geq\frac{1}{\sqrt{k}} \sum_i x_i\approx 1$. 
Since we prove a lower bound against an arbitrary estimator, we can assume that it is given the first column of any sketching matrix $S'=[s_1,S]$ and the structure of the additional row $r_1=[M,M,\ldots ,\sqrt{k}M]$ including the value of $M$. This enables the estimator to remove the influence of the tensor product $s_1\cdot r_1$, such that it can proceed with the estimation on $SZ=S[b,A]$ only. 
Finally, we show that if we solve the above problem (on the sketch) up to a factor of  $1+\Theta(\eps)$, then the resulting solution $\tilde x$ shares a constant fraction of its support with the actual support $U$ of the random codeword $c\in \mathcal C$. Due to the error correcting code construction, this uniquely identifies the \emph{full} unknown support $U$, which concludes the third step of our proof. We hereby obtain our $m=\Omega(k\log(d/k)/\eps^2)$ lower bound against constant error probability oblivious sketching for sparse $\ell_2$ regression.

\paragraph{Sparse $\ell_p$ regression}
Our next aim is to extend the $m=\Omega(k\log (d/k)/\eps^2)$ result to the minimization version of $\ell_p$ norm regression for arbitrary $p\geq 1$.

\begin{restatable}{theorem}{thmLBlp}
\label{thm_inf:LB_lp}
	Let $A\in \R^{n\times d}, b\in \R^n$. Suppose $S\in\R^{m\times n}$ is an oblivious linear sketch for $k$-sparse $\ell_p$ regression for any $p\geq 1$ with an estimator $E_p(SA,Sb,x)$, such that the minimizer $\tilde x \in \argmin_{x\in\Psi_k} E_p(SA,Sb,x)$ satisfies $\|A\tilde x-b\|_p \leq (1+\eps) \min_{x\in \Psi_k} \|Ax-b\|_p$. Then $m=\Omega(k\log(d/k)/\eps^2)$.
\end{restatable}

We start with $\ell_1$ and discuss the general case $p\in [1,\infty)$ below. By Dvoretzky's theorem, we can embed $\ell_2$ into $\ell_1$ with distortion $1\pm\eps$ using a random Gaussian mapping $G$. More precisely, $\|Ax-b\|_2 = (1\pm\eps) \|GAx-Gb\|_1$ for all $k$-sparse $x$, where $G$ has $O(n \log(1/\eps) /\eps^2)$ rows and consists of i.i.d. Gaussians. In particular it is an oblivious embedding. Now suppose we had a sketch $S$ for the $k$-sparse $\ell_1$ regression problem. Then we can show that the minimizer $\tilde x$ for $\|SGAx-SGb\|_1$ is also a $1+O(\eps)$ approximation for $\|Ax-b\|_2$. This implies that $SG$ is an oblivious sketch for the $\ell_2$-norm problem (with an $\ell_1$-norm estimator), and thus $S$ requires $m=\Omega(k \log(d/k) / \eps^2)$ rows.

Similarly, $\ell_2$ embeds obliviously and up to $1\pm\eps$ into $\ell_p$ for all $p\geq 1$ using a random Gaussian matrix $G$ with a number of rows depending on $O(n)$ for $1\leq p\leq 2$ and on $n^{O(p)}$ for $p>2$, \citep[see][p. 30]{Matousek13}, but the number of rows of $G$ does not matter in our context since it is reduced by the sketch $S$. It follows that the $m=\Omega(k \log(d/k) / \eps^2)$ lower bound holds for $\ell_p$, for every $p \geq 1$.

\paragraph{Sparse {ReLU} and hinge-like loss regression}
We further extend the lower bound to the \rm{ReLU} loss function. This time we reduce from the sparse $\ell_1$ regression problem by designing an exact embedding of $\ell_1$ into $\|\cdot\|_{\relu}$.

\begin{restatable}{corollary}{thmLBrelu}
\label{thm_inf:LB_relu}
	Let $A\in \R^{n\times d}, b\in \R^n$. Suppose $S\in\R^{m\times n}$ is an oblivious linear sketch for $k$-sparse $\relu$ regression with an estimator $E_{\relu}(SA,Sb,x)$, such that $\tilde x \in \argmin_{x\in\Psi_k} E_{\relu}(SA,Sb,x)$ satisfies $\|A\tilde x-b\|_{\relu} \leq (1+\eps) \min_{x\in \Psi_k} \|Ax-b\|_{\relu}$. Then $m=\Omega(k\log(d/k)/\eps^2)$.
\end{restatable}

More precisely, it holds for all $x\in \R^d$ that $\|x\|_1=\|x\|_{\relu} + \|-x\|_{\relu}$. A similar argument as in the case of the Gaussian $\ell_2 \rightarrow \ell_p$ embedding yields the lower bound for sketching $\|\cdot\|_{\relu}$, when we replace $G$ by the embedding matrix $P=[I , -I]^T$, which duplicates and negates the input vector. It follows that $SP$ is an oblivious sketch for $\ell_1$ and thus has the same lower bound of $m=\Omega(k\log (d/k)/\eps^2)$.

Finally, in this line of reductions, we deduce a lower bound for a $k$-sparse affine embedding for the class of hinge-like loss functions \citep{mai2021coresets}, which are close to the \rm{ReLU} function up to an additive deviation $a$ (e.g., logistic loss with $a\leq \ln(2)$), see Definition \ref{def_inf:hinge_like}.
\begin{restatable}{definition}{defhingelike}
\label{def_inf:hinge_like}
	We say $f(\cdot)$ is an $(L, a_1,a_2)$ hinge-like loss function if $f$ is $L$ Lipschitz,  $\forall x \geq 0\colon f(x)\geq a_2 >0,$ and $\forall x \colon |f(x) - \relu(x)| \leq a_1.$
\end{restatable}
\begin{restatable}{corollary}{thmLBhingelike}
\label{thm_inf:LB_hingelike}
	Let $A\in \R^{n\times d}, b\in \R^n$. Suppose $S\in\R^{m\times n}$ is an oblivious subspace embedding for some hinge-like loss function $f$ with an estimator $E_f(SA,Sb,x)$, such that we have $\forall x \in \Psi_k\colon  (1-\eps)  \|Ax-b\|_{f} \leq E_f(SA,Sb,x) \leq (1+\eps)  \|Ax-b\|_{f}.$
	Then $m=\Omega(k\log(d/k)/\eps^2).$
\end{restatable}
Since our minimization lower bound for \rm{ReLU} implies an affine embedding lower bound as a direct consequence, it is sufficient to show that any $k$-sparse affine embedding for $f$ yields a $k$-sparse affine embedding for ReLU. The difficulty is that we have an approximate multiplicative embedding of \rm{ReLU} into $f$ only for the positive part. Hence, the remaining part where \rm{ReLU} evaluates to zero needs to be taken care of separately, and the minimization lower bound does not follow directly as in our previous arguments.

\paragraph{Sampling fails for sparse regression}
Sampling based algorithms are important tools in sketching for non-sparse regression \citep{DrineasMM06,DasguptaDHKM09,MunteanuSSW18,mai2021coresets}. However, they do not give any non-trivial results in the sparse setting.

To corroborate why our upper bounds all build on oblivious linear sketches, rather than sampling, we prove lower bounds for matrices $S$ that subsample (and reweight) rows of the input. Here we show that if $A$ is the identity and $b$ is a random standard basis vector, then any algorithm that has access only to $SA$ and $Sb$, must fail with probability $> 1/2$ already in the case $k=1$. 
\begin{restatable}{theorem}{thmSamplingLower}
\label{thm_inf:samplingLower} Consider any bounded approximation factor $\alpha \ge 1$ and any $\norm{\cdot}: \R^n \rightarrow \R^{\ge 0}$ which evaluates to $0$ on the all zeros vector and to some positive number on any other vector. For any $n > 9$, there is some input matrix $A \in \R^{n \times n}$ and distribution over vectors $b \in \R^{n}$ such that for any sampling matrix $S \in \R^{m \times n}$ with $m < n/3$, no algorithm that accesses just $SA$ and $Sb$ can output a $k$-sparse $\tilde x \in \Psi_k$ with $\norm{A\tilde x-b} \le \displaystyle \alpha \cdot \min\nolimits_{x \in \Psi_k} \norm{Ax-b}$ with probability at least $1/2$ (over the choice of $b$ and any possible randomness in the algorithm).
\end{restatable}
The reason is that by the random choice of $b$, if $m < n/3$ then the sketch contains only rows where $b=0$, but to obtain a bounded approximation error, it is crucial to retain the non-zero row. By construction, $A$ cannot help to find this coordinate and thus all possibilities that are not in the sample have equal probability of roughly $1/n$ to succeed. A bound of $m=\Omega(n)$ thus follows.

It is crucial for this bound that the algorithm has no access to $b$ when $S$ is being constructed, supporting the fact that it is the property of \emph{obliviousness} that separates sketching from sampling for sparse regression. However, if our aim is to obtain a $k$-sparse affine embedding via sampling then an $m=\Omega(n)$ bound follows even if the algorithm has full access to the data and even if $k=1, A=I$ and $b=0$, which also means that the RIP property cannot be obtained via sampling.

\begin{restatable}{theorem}{thmSamplingLowerTwo}
\label{thm_inf:samplingLower2} Consider any bounded approximation factor $\alpha \ge 1$ and any $\norm{\cdot}: \R^n \rightarrow \R^{\ge 0}$ which evaluates to $0$ on the all zeros vector and to some positive number on any other vector. For any $n > 9$, there is some input matrix $A \in \R^{n \times n}$ such that there is no sampling matrix $S \in \R^{m \times n}$ with $m < n$, which satisfies for all $k$-sparse $x \in \Psi_k$, $\alpha^{-1} \norm{Ax} \le \norm{SAx} \le \alpha \norm{Ax}$. 
\end{restatable}
We note that in contrast to the impossibility results on sampling, sketching succeeds by relatively simple union bound arguments, e.g., \citep{BaraniukDDW08} for the $\ell_2$ case, as we will see in the next section.

\subsection{Upper Bounds}
\paragraph{Sparse $\ell_2$ regression} Again we start with $\ell_2$. The upper bound is similar to the known constructions \citep{CandesRT06,Donoho06,BaraniukDDW08} of RIP matrices via Johnson-Lindenstrauss embeddings \citep{JohnsL1984}, i.e., Gaussian matrices. The main difference to these works is that the subspaces formed by any fixed $k$-sparse solution space need not be orthogonal or aligned with the standard basis vectors.
\begin{restatable}{theorem}{thmUBltwo}
\label{thm_inf:UB_l2}
	Let $A\in \R^{n\times d}, b\in \R^n$. There exists a distribution over random matrices $S\in\R^{m\times n}$ with $m=O(k\log(d/k)/\eps^2)$ such that it holds with constant probability that
	$\forall x\in\Psi_k: (1-\eps) \|Ax-b\|_2 \leq \|S(Ax-b)\|_2 \leq (1+\eps) \|Ax-b\|_2 .$
\end{restatable}
The idea is that there are at most $\binom{d}{k} \leq (ed/k)^k$ different $k$-sparse supports and each of them corresponds to one choice of $k$ columns of $A$. Every such choice spans a $k$-dimensional linear subspace of $\R^n$.
By the subspace embedding construction of \cite{Sarlo2006}, each subspace formed by one choice of $k$ columns can be handled by embedding the points in a net of size $(3/\eps)^k$ covering the unit ball within the subspace. The remaining vectors can be related to the net points by the triangle inequality and the embedding extends to vectors of arbitrary norm outside the unit sphere by linearity. Indeed, by a more sophisticated argument (see, e.g.,  \citep{Woodruff14}) the net can be constructed with an absolute constant $\eps_0:=1/2$ instead of $\eps$. So the total number of points to embed up to $(1\pm\eps)$ distortion is bounded by $|\calN|\leq (d/k)^k\cdot c^k$ for an absolute constant $c$. The embedding can be accomplished via the Johnson-Lindenstrauss lemma followed by a union bound, by using a Gaussian matrix with $m = O(\log(|\calN|)/\eps^2) = O(k\log(d/k)/\eps^2)$ rows, matching the lower bound.

\paragraph{Sparse $\ell_p$ regression} Towards an extension to $\ell_p$ we focus on $p=1$ first. Unlike the $\ell_2$ case, it is known that a direct embedding of $\ell_1^n$ into $\ell_1^m$ is, even for a fixed constant $\eps$, either exponential, i.e., $m\in\Omega( 2^{\sqrt{d}} )$ or must incur a distortion of $\Omega(d /\polylog(d) )$ \citep{WangW19,LiW021}, and so $(1\pm \eps)$-approximation seems out of reach. Another alternative is the non-linear median estimator of \cite{Indyk2006} which is usually avoided since it leads to a non-convex and usually hard optimization problem in the sketch space. However, we note that $k$-sparse regression is already non-convex and $NP$-hard. So this is a suitable choice in our setting. The sketching matrix of \cite{Indyk2006} is a linear sketch $C\in\R^{m\times n}$ whose entries are i.i.d. Cauchy random variables. Those are known to be $1$-stable, meaning that their dot product with a vector $x$ is again a Cauchy random variable with scale $\|x\|_1$ so that each row yields a reasonable estimator. However, to achieve concentration, the overall estimator is the median of all row estimators instead of their $\ell_1$-norm. This sketch has been combined with a net argument for $\ell_1$ in \citep{BackursIRW16} to obtain a sketching dimension of $O(k\log(k/\eps\delta)/\eps^2)$ for a $k$-dimensional subspace. However, simply taking a union bound over the $\binom{d}{k}$ choices of $k$ columns would result in $m=O((k^2\log(d/k) + k\log(k/\eps))/\eps^2)$, which is far from the lower bound. We open up the proof to improve this to $m=O((k\log(d/k) + k\log(k/\eps))/\eps^2)$, which matches our lower bound unless $k$ is relatively large (in the order of $k = \omega(\sqrt{\eps d})$).
\begin{restatable}{theorem}{thmUBlonep}
\label{thm_inf:UB_l1_lp}
	Let $A\in \R^{n\times d}, b\in \R^n, p\in [1,2)$. There exists a distribution over random matrices $S\in\R^{m\times n}$ with $m = O(k(\log(d/k) + \log(k/(\eps\delta)))/\eps^2)$ such that it holds with probability at least $1-\delta$ that
	$\forall x\in\Psi_k: (1-\eps)\normm{p}{Ax-b} \leq \normm{\rm med}{S(Ax-b)} \leq (1+\eps) \normm{p}{Ax-b},$
	where for arbitrary $y\in \R^d, \normm{\rm med}{y}:= \operatorname{median}\{|y_i| \mid i\in[d]\}$.
\end{restatable}
We obtain similar upper bounds by generalizing this result to $\ell_p, p\in[1,2)$. The sketching matrix is again a linear sketch $C\in\R^{m\times n}$ whose entries are i.i.d. random variables drawn from a $p$-stable distribution, generalizing the $1$-stable Cauchy distribution. Such an extension has been proposed by \cite{Indyk2006} for sketching single vectors but to our knowledge has never been worked out due to the lack of closed form expressions for the cumulative density function (cdf) and probability density function (pdf), except for $p\in\{1,2\}$. Here we show how to obtain directly a subspace embedding for all $k$-sparse vectors. To this end we leverage bounds on the tails of $p$-stable distributions \citep{BednorzLM18}. We note that $p$-stable distributions are leptokurtic. More specifically, they are heavy-tailed with decay $\Pr[|X| > \tau]\leq 1/\tau^p$ except for $p=2$. Therefore we need to rely on a non-linear quantile estimator to achieve concentration for any $p\in[1,2)$ to construct and apply the net argument, as in the case $p=1$ described above. A more intriguing question is how to analyze the cdf of $p$-stables without closed form expressions. Our solution to this problem comes from the fact that the characteristic function of any $p$-stable distribution is known in closed form and equals the pdf of a $p$-generalized normal distribution up to a normalizing constant \citep{Dytso18}. I.e., it is given by $\phi(t)=\exp(-|\gamma_p t |^p)$, where $\gamma_p$ is a constant scale parameter that depends on $p$. Using an inversion theorem of L\'{e}vy, we can analyze the cdf and its derivative via an integral involving the characteristic function. As a side result we affirm a conjecture by \cite{Indyk2006} that a $(1\pm\eps)$ approximation can be obtained via the median estimator for all $p\in [1,2]$\footnote{The reference \citep{Indyk2006} gives a non-constructive proof showing that there exists some (unknown) quantile, possibly depending on $p$ and $\eps$, that yields a good estimator.}.

\paragraph{ReLU and hinge-like loss functions}
It is well-known \citep{MunteanuSSW18} that without any assumptions, these types of functions do not admit relative error sketches with $o(n)$ rows. To address this issue, \cite{MunteanuSSW18} introduce a natural notion for the complexity of sketching the matrix $A$, which we also use to parameterize our results.
Intuitively, the parameter $\mu:=\sup_{x\in \mathbb R^d\setminus \{0\}}\normm{1}{(Ax)^+}/\normm{1}{(Ax)^-}$ is large when there is some $x$ that produces a significant imbalance between the $\ell_1$-norm of all positive and the $\ell_1$-norm of all negative entries. This can occur, e.g., when the data admits perfect linear separability. However, as \cite{MunteanuSSW18} argued, we typically expect $\mu$ to be small.
These assumptions were recently leveraged to develop the first oblivious linear sketches for logistic regression \citep{MunteanuOW21,MunteanuOW23}, which led to an efficient algorithm for the minimization problem in the non-sparse regime. 
Another recent contribution of \cite{mai2021coresets} led to small dependencies on $d,\mu$ and $\eps$ in the regime of sampling and coreset algorithms. However, we already argued that sampling would not work in our setting of sparse regression.

We therefore combine and extend those results to \emph{sketching} a wider class of loss functions and with better dependencies on the approximation parameters, based on the median sketch for $\ell_1$. We note that this sacrifices the efficiency of optimization in the sketch space, but as we have argued before, in the context of sparse regression, finding the right support is already a hard problem, which motivates us to focus on the best possible parameterization.
\begin{restatable}{theorem}{thmUBrelufromlone}
	\label{thm_inf:UB_relu_from_l1}
	Let $A\in \R^{n\times d}, b\in \R^n$. There exists an oblivious sketch $S$ with $O(\frac{\mu^2 k}{\epsilon^2} \log(\frac{\mu  d}{\epsilon \delta} ))$ rows and an estimator $g_{\relu}(SA,Sb,x)$, such that with probability at least $1 - \delta$, we have $\forall x\in\Psi_k: (1-\epsilon) \norm{Ax - b}_{\relu}\leq g_{\relu}(SA,Sb,x) \leq (1 + \epsilon) \norm{Ax - b}_{\relu}.$
\end{restatable}
Our $k$-sparse affine embedding sketch leverages the fact that $\relu(x)=({\sum_{i} x_i + \|x\|_1})/{2}$, since the negative entries are contained negatively in the sum and again positively in the norm, so they cancel. The positive values are positive in both parts and thus counted twice, so dividing by $2$ yields the exact value of $\relu(x)$. The sum of entries can be sketched exactly using only one row vector and the $\ell_1$ norm is sketched via the ($1$-stable) Cauchy sketch with median estimator, as detailed above in the previous paragraph. Now, the error of this estimate is $\eps\|x\|_1$ but by the $\mu$-complexity assumption, the $\ell_1$ norm is within roughly a $\mu$-factor of the positive entries, so folding $\mu$ into $\eps$ yields an error of $\eps\relu(x)$.
We extend this result to an even richer class of hinge-like loss functions, including logistic regression. Those functions are additively close to the \rm{ReLU} function, see Definition \ref{def_inf:hinge_like}. The logistic regression loss $\zeta(x)=\ln(1+\exp(x))$, for instance, has asymptotes equal to \rm{ReLU} in the limit of $\pm\infty$. However, close to zero, the two functions differ more significantly, attaining a bounded maximum deviation of $\zeta(0)-\relu(0)=\ln(2)$.
\begin{restatable}{theorem}{thmUBhingelikefromrelu}
	\label{thm_inf:UB_hingelike_from_relu}
	Let $A\in \R^{n\times d}$, $b\in \R^n$, and let $f$ be an $(L, a_1, a_2)$ hinge-like loss function. There exists an oblivious sketch $S$ with  
	$m = O( \frac{c^{10}\mu^2k}{\epsilon^2} \log ( \frac{cn \mu d }{\epsilon \delta} ) )$ rows, where $c= \max(1,L,a_1,1/a_2)$, and an estimator $g_f(SA,Sb,x)$, such that, with probability at least $1 - \delta$, we have $\forall x\in\Psi_k: (1-\epsilon) \norm{Ax-b}_{f}\leq g_{f}(SA,Sb,x) \leq (1 + \epsilon) \norm{Ax-b}_{f}. $
\end{restatable}
The idea is now to split the loss function into two components $\sum_i (\zeta(x_i)-\relu(x_i)) + \sum_i \relu(x_i)$. The \rm{ReLU} function can be dealt with as described in the paragraph above and the remainder is a sum over bounded terms. This again enables us to achieve concentration and union bound over the net of $k$-sparse vectors up to an additive error of roughly $\eps \frac{n}{\mu}$. This error can be charged by a complementing lower bound which follows from leveraging the $\mu$-complexity assumption, akin to \citep{mai2021coresets,MunteanuOW21}, and finally yields a relative error guarantee.

For completeness, we have the following simple result that yields a connection between minimizing in the sketch space of a $k$-sparse affine embedding, as in all upper bounds above, and the original minimization problem.

\begin{restatable}{corollary}{corfminimization}
\label{cor_inf:f_minimization}
	Let $A\in \R^{n\times d}, b\in \R^n$. Let $\norm{\cdot}:\R^n \rightarrow \R^{\geq 0}$ be any loss function. Let $S$ be an oblivious linear sketch of $[A,b]$, and let $E(SA,Sb,x)$ be an estimator that satisfies
	$\forall x\in\Psi_k: (1-\eps) \norm{Ax-b} \leq E(SA,Sb,x) \leq (1+\eps) \norm{Ax-b} .$
	Then $\tilde{x} \in \argmin_{x \in \Psi_k} E(SA,Sb,x)$ satisfies
	$\|A\tilde{x}-b\| \leq (1 + O(\eps)) \min_{x \in \Psi_k} \norm{Ax-b}.$
\end{restatable}
\subsection{LASSO Regression}
LASSO regression \citep{Tibshirani96} is a convex relaxation of $k$-sparse $\ell_2$ regression and enjoys large popularity as a heuristic for inducing sparsity and feature selection. LASSO regression is a special subject of our investigation, since here we do not assume that any solution is $k$-sparse. In this case, a subspace embedding for the dense $\ell_2$ problem is possible for sketching down to $m=\Theta(d/\eps^2)$ dimensions \citep{Sarlo2006,NelsonN14}. Also in cases where the regularization parameter is very small and thus the norm of the solution is actually unconstrained, the problem becomes equivalent to least squares regression, in which case $\Theta(d/\eps)$ is necessary and sufficient \citep{Sarlo2006,ClarkW09}. Usually, however, the $\ell_1$ regularization is imposed to yield a sparse minimizer, for which we can again hope to be able to take advantage of the induced sparsity, parameterized by the value of the regularization parameter $\lambda$ such as to reduce to $\operatorname{poly}(1/\lambda,\log d)$ rows.
Here we give an upper bound for sketching that depends on an $\ell_1$ regularization parameter $\lambda$.
\begin{restatable}{theorem}{thmLASSOUpper}
\label{thm_inf:LASSOUpper}
	Consider $A \in \R^{n \times d}$,  $b \in \R^{n}$, and $\lambda \in (0,1)$. Assume that $\norm{A}_2 \le 1$ and $\norm{b}_2 \le 1$,  If $S \in \R^{m \times n}$ is a random Gaussian matrix (i.e., each entry is sampled i.i.d. from ${N}(0,1/m)$) then for any $\epsilon,\delta \in (0,1)$ and $m = O(\frac{\log d/\delta}{\lambda^2 \cdot \epsilon^2}) $, with probability at least $1-\delta$, if $\tilde x = \argmin_{x \in \R^d} \norm{SA x - Sb}_2^2 + 
	\lambda \norm{x}_1$ then $\norm{A \tilde x - b}_2^2 + \lambda \norm{\tilde x}_1 \le (1+\epsilon) \cdot \min_{x \in \R^d} \norm{Ax - b}_2^2 + \lambda \norm{x}_1.$
\end{restatable}
Observe that our constraints on $\norm{A}_2,\norm{b}_2$ are necessary, since LASSO regression is \emph{not} scale invariant.
The result follows by first showing that the optimizer must have a bounded norm in terms of the optimal objective value $\norm{x}_1 \le \frac{2 \cdot OPT}{\lambda }$. This allows us to focus on the set $\mathcal{T} = \{y = Ax-b: \norm{x}_1 \le \frac{2 \cdot OPT}{\lambda }\}$, which by the bounded norm, can be expressed as the convex hull of $2d+1$ points in the unit ball. For this set, we can use an embedding result of \cite{NaN18} to obtain an additive error of $\epsilon \lambda$ for all vectors in the set $\mathcal{T}$, which allows us to relate the sketching error to $O(\eps \cdot OPT)$, and which finally yields our $1+\eps$ relative error approximation result.

We complement the upper bound by the following lower bound that matches the dependence on $d$ and $\lambda$. 
The proof builds on our new techniques developed for sparse $\ell_2$ regression. The condition of bounded norm inputs $A,b$, however, does not allow us to plant the additional row gadget. Thus, we need to choose a smaller $\lambda$ by a factor of $\epsilon$, which unfortunately cancels the $\epsilon$ dependence in the previous $\ell_2$ lower bound. Still, our result shows that $\log(\lambda d)/\lambda^2$ rows are necessary for any sketch with an estimator that allows to solve LASSO to within a $1+\eps$ approximation.
\begin{restatable}{theorem}{thmLASSOLower}
\label{thm_inf:LASSOLower}
	Let $A\in \R^{n\times d}, b\in \R^n$ with bounded $\norm{A}_2 \le 1$ and $\norm{b}_2 \le 1$, and let $\lambda \in (0,1)$. Suppose $S\in\R^{m\times n}$ is an oblivious linear sketch for LASSO regression with an estimator $E_\lambda(SA,Sb,x)$, such that with constant probability 
	$\tilde x \in \argmin_{x\in\Psi_k} E_\lambda(SA,Sb,x)$
	satisfies
	$\norm{A \tilde x - b}_2^2 + \lambda \norm{\tilde x}_1 \le (1+\epsilon) \cdot \min_{x \in \R^d} \norm{Ax - b}_2^2 + \lambda \norm{x}_1.$
	Then $m = \Omega(\frac{\log( \lambda d )}{\lambda^2 } ) .$
\end{restatable}

\subsection{Separation of Sparse Recovery from Sparse Regression}
Here we give an upper bound of $m=O(k \log(d)/\eps + k\log(k/\eps))/\eps^2)$ that gets very close to the lower bound $m=\Omega(k\log(d/k)/\eps+k/\eps^2)$ of \cite{price20111+}. Surprisingly, this provides a separation between the $k$-sparse recovery problem and the $k$-sparse regression problem. Combined with our main $\Omega(k\log(d/k)/\eps^2)$ lower bound, it shows that $k$-sparse regression is \emph{strictly harder} to sketch than sparse recovery. To obtain the new upper bound, the issue is that we need to figure out which subset to use, but cannot afford to estimate every subset’s cost, even though for each subset we can get a good estimate. To cope with this problem we run in parallel a two-stage estimation procedure: one CountSketch that gives a rough estimate of the entries, which yields a superset $I$ of the $k$ coordinates of interest, i.e., $I\subseteq [d]$ of size $|I|=O(k/\eps)$; and another CountSketch that has a higher precision, but is required only to recover estimates for the relatively small number of $k$-subsets of the set $I$. Now, if we output the vector that is supported on the top $k$ entries in $I$ together with their estimates of the entries obtained from the high precision sketch, this gives the desired $1+\eps$ approximation for $k$-sparse recovery with fewer rows than necessary to solve the $k$-sparse regression problem.

\begin{restatable}{theorem}{thmUBsparserecovery}
\label{thm_inf:UB_sparse_recovery}
	On input $x\in \mathbb{R}^d$, the above sparse recovery scheme uses $O(k\log(k/\epsilon)/\epsilon^2 + k \log(d)/\epsilon)$ measurements and, with probability at least $1-1/\poly(d)-1/\poly(k/\epsilon)$, returns a $k$-sparse
	vector $\hat{x}\in\Psi_k$ satisfying 
	$\|x - \hat{x}\|_2^2 \leq (1+\epsilon) \min_{x_k\in\Psi_k}\|x-x_k\|_2^2.$
\end{restatable}

\section{Conclusion}
In this paper we study the complexity of oblivious linear sketching for sparse regression problems under various regression loss functions such as $\ell_p$ regression, logistic regression, ReLU loss, and hinge-like loss functions. Our results are \emph{essentially}\footnote{up to minor polylogarithmic terms and problem specific parameters such as $\mu$; see \cref{tab:results} for the exact terms.} tight bounds of $\Theta(k\log(d/k)/\eps^2)$ for all those problems. We further study the sketching complexity of LASSO, a popular convex relaxation often used as a heuristic for solving sparse linear regression. We give the first bound of $O(\log(d)/(\lambda\eps)^2)$ going below the linear dependence on $d$, where $\lambda$ is the regularization parameter.
Furthermore we provide a separation result from the sparse recovery problem studied in compressed sensing. Surprisingly, we find that the sparse regression problem requires $m=\Omega(k \log(d/k)/\eps^2)$ and is thus strictly harder to sketch than sparse recovery, for which we show a new $m=O(k \log(d)/\eps + k\log(k/\eps)/\eps^2)$ upper bound. We also show that while data dependent importance sampling techniques are widely successful for the unconstrained non-sparse regression problems, they do not give any non-trivial bounds in the sparse setting. This underlines the importance of oblivious sketching techniques in the sparse context.
For future directions we aim at closing remaining gaps, especially for hinge-like loss functions. It will also be an interesting avenue to develop more scalable and faster heuristics by incorporating our sketching techniques and evaluate their performance in practice. Finally, since our sketches are optimized for a smallest possible target dimension, it will be interesting to study the trade-off between the speed of applying them to data and an increase in their target dimension.

\section*{Acknowledgements}{Alexander Munteanu was supported by the German Research Foundation (DFG), Collaborative Research Center SFB 876, project C4 and by the Dortmund Data Science Center (DoDSc). D. P. Woodruff was supported in part by Office of Naval Research (ONR) grant N00014-18-1-2562 and a Simons Investigator Award. Cameron Musco was supported by an Adobe Research grant, a Google Research Scholar Award, and NSF Grants No. 2046235 and No. 1763618}

\clearpage
\printbibliography

\onecolumn
\appendix
\section{Lower bounds for \texorpdfstring{$k$}{k}-sparse regression}\label{sec:LB}
\subsection{Lower Bounds for the \texorpdfstring{$\ell_p$}{ℓₚ}-norm Loss Function for \texorpdfstring{$p\geq 1$}{p ≥ 1}}\label{sec:LB_lp}

In this section we prove our main \cref{thm_inf:LB_l2} on $\ell_2$ followed by our extension to $\ell_p, p\geq 1$, see \cref{thm_inf:LB_lp} below.
\thmLBltwo*
\noindent
The outline is as follows:
\begin{enumerate}
    \item We construct a suitable (hard) distribution over $k$-sparse supports, which is used to define our input distribution. 
    \item We prove the impossibility of recovering a constant fraction of the support with a small number of measurements (rows) from the input distribution below an information-theoretic lower bound.
    \item We construct an $\ell_2$-regression instance for which any $1 + \Theta(\eps)$ approximation derived from an oblivious sketch, paired with an arbitrary estimator, reveals a constant fraction of the support. The hardness result thus turns over to the regression problem.
\end{enumerate}

For the first proof step, we begin with the construction of an error correcting code, which will be used in the main argument to construct a hard input distribution.

\begin{definition}[Balanced and Correctable Support Set]\label{def:bal} Consider a set $\mathcal U \subset [d]^{k+1}$ of sets of $k+1$ indices, such that for all $U \in \mathcal U$, $1 \in U$.
$\mathcal U$ is said to be balanced if, letting $c_i = |\{U \in \mathcal U: i \in U\}|$, we have $c_i = c_j$ for all $i,j \in [d] \setminus \{1\}$, and further, letting $c_{ij} =  |\{U \in \mathcal U: i,j \in U\}|$, $c_{ij} = c_{kl}$ for all $i \neq j$ and $k \neq l$ with $i,j,k,l \in [d]\setminus \{1\}$. The set is said to be correctable if for all $U_1,U_2 \in \mathcal U$, $|U_1 \cap U_2| \le 9/10 \cdot k$.
\end{definition}

We prove the existence of a suitably large balanced and correctable support set as follows.
Suppose we choose $t$ sets $S_1, \ldots, S_t$, each of size $k$, uniformly at random, and each from $[n] = \{1, 2, \ldots, n\}$. Let $\mathcal{H}$ be a pairwise independent family of $n \cdot (n-1)$ hash functions $h:[n] \rightarrow [n]$; it is well-known that such a family exists when $n$ is prime \citep{carter1979universal}, which we can assume without loss of generality. For each $S_i$ and $h \in \mathcal{H}$, let $h(S_i)$ denote the image of $S_i$ under $h$. 

\begin{lemma}\label{lem:set_construction}
For any constant $0 < c < 1$, there exists a constant $C > 0$ and $t = \exp(C k \log(n/k))$ subsets $S_1, \ldots, S_t$, each of size $k$, such that $|h(S_i) \cap h'(S_j)| < c k$ for all $1 \leq i < j \leq t$ and all $h \neq h' \in \mathcal{H}$.
\end{lemma}
\begin{proof}
Note that for $i < j$, $h(S_i)$ and $h'(S_j)$ are each random and independent subsets of size $k$. To calculate their intersection size, we can fix $h(S_i)$. Then the probability that $|h(S_i) \cap h'(S_j)|$ is at least $ck$ is at most the probability that some subset of $h(S_i)$ of size $ck$ is also a subset of $h'(S_j)$. This probability is in turn bounded by
\begin{eqnarray*}
	{\binom{k}{ck}} \cdot \frac{{\binom{n-ck}{k-ck}}}{{\binom{n}{k}}} 
	& \leq & {\binom{k}{ck}} \cdot \frac{(n-ck)! \cdot k!}{(k-ck)! \cdot n!}\\
	& \leq & 2^k \cdot \left (\frac{k}{n} \right )^{ck}
	 \leq  \exp(-C k \log(n/k)),
\end{eqnarray*}
where $C > 0$ is a suitable constant. 
Consequently, by a union bound over all pairs $1 \leq i < j \leq t$ and all choices of $h \neq h' \in \mathcal{H}$, we conclude there
exists a choice of $t = \exp(C k \log(n/k))$ such sets, for a
different choice of constant $C > 0$. 
\end{proof}

The lemma above implies that if we have a set $T$ of size $k$ which intersects some $h(S_i)$ in at least $ck$ positions, then $S_i$ is uniquely determined.

Moreover, by pairwise independence of $\mathcal{H}$, for any $a \neq b \in [n],$ the number of sets in the union $\cup_{h \in \mathcal{H}, 1 \leq i \leq t} h(S_i)$ for which $a$ and $b$ occur together is the same. Also, the number of sets in $\cup_{h \in \mathcal{H}, 1 \leq i \leq t} h(S_i)$ containing any particular value $a \in [n]$ is the same as for any other particular value $b \in [n]$. 
Using this argument we construct the desired set by simply applying Lemma \ref{lem:set_construction} with $k$ and $n = d$ for indexing a family of sets over $\{2,\ldots,d+1\}$ and append the element $1$ to each of the sets. The hard instance will later be constructed from a uniform element of this set. This concludes the first part of our proof.

As a second step, we prove the impossibility of recovering a constant fraction of the support with a small number of measurements (rows) from the input distribution below an information-theoretic lower bound. To this end, we need to bound the mutual information first and then plug it into Fano's inequality.

\begin{lemma}[Mutual Information Bound]\label{lem:mutInfo}
Let $U$ be selected uniformly at random from a balanced and correctable support set $\mathcal{U}$ (Def. \ref{def:bal}) and for some $\epsilon > 0$, let $z \in \R^{d}$ have $z(1) = 1$, $z(i) = \epsilon/\sqrt{k}$ for all $i \in U \setminus \{1\}$ and $z(i) = 0$ for all $i \notin U$. Let $X \in \R^{n \times d}$ have rows drawn independently from a mean zero multivariate Gaussian distribution with covariance $I+zz^T$. Then:
\begin{align}\label{eq:logdet}
	I(U;X) \le 7n\cdot \epsilon^2.
\end{align}
\end{lemma}
\begin{proof}
Starting from the high level outline of \cite{AminiW2009}, since the rows of $X$ are independent, we can write:
\begin{align}\label{eq:mi}
	I(U;X) &= H(X)-H(X|U)\nonumber\\
	&\le n \cdot [H(x) - H(x|U)] = \frac{n}{2} \cdot \left [\log\det(\E[xx^T]) - \log\det(\E[xx^T | U])\right].
\end{align}
We now compute the needed log determinants.
First observe that $\E[xx^T | U] =  \E[I + zz^T | U] = I + \E[zz^T | U]$. We can observe that $\E[zz^T | U] = D MD$ where $M_{ij} = \epsilon^2/k$ for $i,j \in U$ and $M_{ij} = 0$ otherwise, and where $D$ is diagonal all with $D_{11} = \sqrt{k}/\epsilon$ and $D_{ii} = 1$ for all $i \neq 1$. Observe that $M$ is rank-$1$ and positive semidefinite. Thus, so is $DMD$. Thus $DMD$ has one non-zero eigenvalue, equal to its trace, which is $1+\epsilon^2/k \cdot k = 1+\epsilon^2$. Thus, $\E[xx^T | U] = I + DMD$ has one eigenvalue equal to $2 +\epsilon^2$ and $n-1$ eigenvalues equal to $1$, so 
\begin{align}\label{eq:logDetNoCond}
	\log\det(\E[xx^T | U]) = \log(2+\epsilon^2) \ge \log(2).
\end{align}

Next consider $\E[xx^T]$. Again we have $\E[xx^T] = \E[I + zz^T] = I + \E[zz^T]$. Since $\mathcal{U}$ is balanced, for all $i \in [d]\setminus \{1\}$, $z(i) = \epsilon/\sqrt{k}$ with probability $k/(d-1)$ and for $i,j \in [d]\setminus \{1\}$ with $i \neq j$, $z(i) = z(j) = \epsilon/\sqrt{k}$ with probability $\frac{k(k-1)}{(d-1)(d-2)}$. We can write $\E[zz^T] = D + E$. Here, $E_{11} = 0$, $E_{ij} = \frac{\epsilon^2}{k} \cdot \frac{k(k-1)}{(d-1)(d-2)} = \frac{\epsilon^2(k-1)}{(d-1)(d-2)}$ for $i, j \neq 1$, and $E_{i1} = E_{1i} = \frac{\epsilon}{\sqrt{k}} \cdot \frac{k}{d-1} = \frac{\epsilon \sqrt{k}}{d-1}$ for $i \neq 1$. $D_{11} = 1$, $D_{ii} = \frac{\epsilon^2}{d-1}- \frac{\epsilon^2}{k} \cdot \frac{k(k-1)}{(d-1)(d-2)} \le \frac{\epsilon^2}{d-1}$. 

Observe that $\norm{E}_F \le \sqrt{2(d-1)\cdot \frac{\epsilon^2 k}{(d-1)^2}+(d-1)^2 \cdot \frac{\epsilon^4(k-1)^2}{(d-1)^2(d-2)^2}} \le \sqrt{3} \cdot \epsilon^2$. Further, $E$ is rank-$3$ and thus has just $3$ non-zero eigenvalues. By Weyl's inequality, $$\lambda_1(\E[zz^T] ) = \lambda_1(D+E) \le \lambda_1(D)+\lambda_1(E) \le 1 + \sqrt{3} \epsilon^2.$$ For $i =2,3,4,5$, $$\lambda_i(\E[zz^T] ) = \lambda_i(D+E) \le \lambda_2(D) + \lambda_1(E) \le \frac{\epsilon^2}{d-1} + \sqrt{3} \epsilon^2 \le (\sqrt{3}+1) \epsilon^2.$$ Finally, for $i \ge 5$, $$\lambda_i(\E[zz^T] ) = \lambda_i(D+E) \le \lambda_2(D)+\lambda_4(E) = \frac{\epsilon^2}{d-1} + 0 = \frac{\epsilon^2}{d-1}.$$
Thus, 
\begin{align}\label{eq:logDetNoCond2}\log\det(\E[xx^T]) &= \log\det(I + \E[zz^T])\nonumber\\ 
	&\le \log(2+\sqrt{3} \epsilon^2) + 4 \log(1+(\sqrt{3}+1) \epsilon^2) + (d-5) \log \left (1+\frac{\epsilon^2}{d-1} \right )\nonumber\\
	&\le \log(2) + (5\sqrt{3}+4+1)\epsilon^2 \nonumber\\
	&\le \log(2)+14\epsilon^2.
\end{align}
Combined with \eqref{eq:logDetNoCond} we have $\log\det(\E[xx^T]) - \log\det(\E[xx^T | U]) \le 14 \epsilon^2$, and plugging back into \eqref{eq:mi}, we have $I(U;X)  \le 7n \cdot \epsilon^2$ as desired.
\end{proof}

Our mutual information bound can be plugged into Fano's inequality to obtain a lower bound on the sample complexity needed for an approximate, i.e., partial, recovery of the support set. This will later translate into the number of rows of our sketch.

\begin{corollary}[Sample Complexity Lower Bound]\label{cor:fano}
Let $U,X$ be distributed as in Lemma \ref{lem:mutInfo} with $n \le \frac{c k \log(d/k)}{\epsilon^2}$ for sufficiently small constant $c$. Then no algorithm that takes just $X$ as input can output a set $\tilde U$ with $|\tilde U| = k$ and $|\tilde U \cap U| > 19k/20$ with probability $\ge 2/3$.
\end{corollary}

\begin{proof}
Suppose such an algorithm existed. Since all $U,U' \in \mathcal{U}$ have $|U \cap U'| \le k \cdot 9/10$, if $|\tilde U \cap U| >19k/20$, then $\tilde U$ must contain $> k/20$ elements not in $U'$ for any $U' \in \mathcal U$ with $U' \neq U$. Thus, we must have $|\tilde U \cap U'| <19k/20$. So $\tilde U$ can be used to uniquely identify $U$. That is, the algorithm identifies $U$ with probability $\ge 2/3$. However, by Fano's inequality (Lemma \ref{lem:fano}), the algorithm fails with probability at least 
$$1- \frac{I(U;X)+ \log 2}{\log|\mathcal{U}|} \ge 1- \frac{7\epsilon^2 n+\log 2}{\log|\mathcal{U}|},$$
where the bound on $I(U;X)$ follows from Lemma \ref{lem:mutInfo}. Since $\log|\mathcal{U}| = \Theta(k \log(d/k))$ this failure probability is $> 1/3$ if $n = \frac{c k \log(d/k)}{\epsilon^2}$ for small enough $c$ and $d,k$  are bigger than large enough constants. This gives a contradiction to the assumption that the algorithm succeeds with probability $\ge 2/3$, and hence the corollary.
\end{proof}

This concludes the second part of our proof regarding the hardness of support recovery. For the third part, i.e., the reduction of this hard problem to sparse linear regression, we first need a few technical lemmas, before we can finally prove \cref{thm_inf:LB_l2}.

The first technical result establishes a connection between approximating a loss function $L$ to within $(1+O(\eps))$ error and revealing a constant fraction of the support. We note that $L$ will represent the regression cost in our subsequent reduction.
\begin{lemma}
	\label{lem:optimize_L}
	Let $v \in \R^d$ be a $k$-sparse vector with $k$ non-zero entries equal to $1/\sqrt{k}$. Let $x$ be another $k$-sparse vector. Let $M > \sqrt{n/(\epsilon k)}$, $\alpha = |\supp(v) \cap \supp(x)|$, and  
	\[
	L = 1 + \|{x}\|_2^2 + (1 - \epsilon {x}^T {v})^2 + \frac{M^2}{n} \left(\sum_i x(i) - \sqrt{k}\right)^2.
	\]
	There exists a constant $c$ such that any $x$ with $\alpha < 19k/20$ is not a $1 + c\epsilon$ approximation solution of $L$. 
\end{lemma}

\begin{proof}
	Let $S = \supp({x}) \cap \supp({v})$ and $M' = M/\sqrt{n}$. Let $\beta = \sum_i x(i)$ and $\gamma$ be such that $\sum_{i \in S} x(i) = \gamma \beta$. We will optimize $L$ over all possible values of $\beta$ and $\gamma$ in $\mathbb{R}$.
	Note that any $x$ minimizing $L$ must have the form 
	\[
	x(i) = 
	\begin{cases}
		\frac{\gamma \beta}{\alpha} &\text{ for } i \in S, \\
		\frac{(1-\gamma)\beta}{k - \alpha}  &\text{ for } i \in \supp(x) \setminus S.
	\end{cases} 
	\]
	This is because for fixed $S, \beta,\gamma$, making $x$ have the above form minimizes $\norm{x}$ without affecting $x^Tv$.
	Therefore, 
	\begin{align*}
		L &= 1 + \frac{\gamma^2 \beta^2}{\alpha} + \frac{(1-\gamma)^2\beta^2}{k - \alpha} + \left( 1 - \epsilon \alpha \frac{1}{\sqrt{k}} \frac{\gamma \beta}{\alpha} \right)^2 + M'^2 \left(\beta - \sqrt{k}\right)^2 \\
		&= \gamma^2 \left( \frac{\beta^2}{\alpha} + \frac{\beta^2}{k - \alpha} + \epsilon^2 \frac{\beta^2}{k}\right) 
		- 2 \gamma \left( \frac{\beta^2}{k-\alpha} + \epsilon \frac{\beta}{\sqrt{k}}\right) 
		+ 2 + \frac{\beta^2}{k - \alpha} + M'^2(\beta - \sqrt{k})^2.
	\end{align*}
	Minimizing over all $\gamma \in \mathbb{R}$ gives 
	
	\begin{align*}
		\min_{\gamma} L &= 2 + \frac{\beta^2}{k - \alpha} + M'^2(\beta - \sqrt{k})^2 - \frac{\left( \frac{\beta^2}{k-\alpha} + \epsilon \frac{\beta}{\sqrt{k}}\right)^2}{\frac{\beta^2}{\alpha} + \frac{\beta^2}{k - \alpha} + \epsilon^2 \frac{\beta^2}{k}} \\
		& = 2 + \frac{\beta^2}{k - \alpha} + M'^2(\beta - \sqrt{k})^2 - \frac{\frac{\beta^2}{(k -\alpha)^2}  + \frac{2\beta \epsilon}{(k -\alpha)\sqrt{k}} + \frac{ \epsilon^2}{k}}{\frac{k}{\alpha(k-\alpha)} + \frac{\epsilon^2}{k}} \\
		&= 2 + M'^2(\beta - \sqrt{k})^2 - \frac{\frac{\beta^2}{(k -\alpha)^2}\left( 1 - \frac{k}{\alpha}\right) + \frac{\beta^2 \epsilon^2}{k(k -\alpha)}  + \frac{2\beta\epsilon}{(k -\alpha)\sqrt{k}} + \frac{ \epsilon^2}{k}}{\frac{k}{\alpha(k-\alpha)} + \frac{\epsilon^2}{k}} \\
		&= 2 + M'^2(\beta - \sqrt{k})^2 + \frac{\beta^2k - \beta^2\epsilon^2 \alpha - 2 \beta \sqrt{k} \alpha  \epsilon - \epsilon^2(k-\alpha)\alpha}{k^2 + \epsilon^2(k-\alpha)\alpha} \\
		&\approx_{\epsilon} 2 + M'^2(\beta - \sqrt{k})^2 + \frac{\beta^2 k - 2 \beta \sqrt{k} \alpha \epsilon}{k^2} \\
		&= \beta^2\left( M'^2 +  \frac{1}{k}\right) - \beta \left( 2 M'^2 \sqrt{k} + \frac{2 \alpha \epsilon}{k \sqrt{k}} \right) + 2 + M'^2k. 
	\end{align*} 
	Minimizing over $\beta$ gives
	\begin{align*}
		\min_{\beta }\min_{\gamma} L &\approx_{\epsilon} 2 + M'^2k - \frac{\left( M'^2\sqrt{k} + \frac{\alpha \epsilon}{k \sqrt{k}}\right)^2}{M'^2 + \frac{1}{k}} \\
		&= 2 + \frac{M'^2 - \frac{2M'^2\alpha \epsilon}{k} - \frac{\alpha^2 \epsilon^2}{k^3}}{M'^2 + \frac{1}{k}} \\
		&\approx_{\epsilon} 2 + \frac{M'^2 - \frac{2M'^2\alpha \epsilon}{k}}{M'^2 + \frac{1}{k}}.
	\end{align*}
	Since $M' = M /\sqrt{n} > 1/\sqrt{\epsilon k}$, we have 
	\[
	\min_{\beta }\min_{\gamma} L \approx_{\epsilon} 3 - 2 \frac{\alpha}{k} \epsilon.
	\]
	We can observe that the RHS is minimized when $\alpha = k$ at $3 - 2\epsilon$. Moreover, if $\alpha < c_1 k$ for $c_1 < 1$, it is at least 
	$$
	3 - 2 c_1 \epsilon \geq (3-2\epsilon)\left(1 + \frac{2(1-c_1)}{3}\epsilon \right).
	$$
	Therefore, there exists a small enough constant $c$ such that if $\alpha < {19k}/{20}$, it is not possible to minimize $L$ within $1 + c \epsilon$ factor.
\end{proof}

The next ingredient will help us analyze the regression cost up to $1\pm\eps$ error deterministically by removing the influence of a random Gaussian matrix in our input distribution.
By a standard tail bound for Gaussian matrices (for example, Exercise 4.7.3 in \cite{vershynin2018high}) we have the following lemma.
\begin{lemma}
	\label{lem:Gaussian_concentration}
	Suppose $X$ is a $n \times d$ Gaussian matrix with covariance $\Sigma.$ Then there exists a constant $C$ such that we have with probability $\geq 1 - \delta$
	\[
	(1- \epsilon) \Sigma \preceq \frac{1}{n} X^T X \preceq (1+\epsilon) \Sigma,
	\]
	when $n \geq C \frac{d + \log(1/\delta)}{\epsilon^2}.$
\end{lemma}
By invoking the above lemma to a fixed $k$-dimensional subspace and applying a union bound over $\binom{d}{k}$ $k$-dimensional subspaces, we have with probability at least $1 - \delta$
\[
\forall k\mbox{-sparse vectors } v\in \Psi_k\colon (1- \epsilon) v^T \Sigma v \preceq \frac{1}{n} v^T X^T X v \preceq (1+\epsilon) v^T\Sigma v,
\]
when $n \geq C' \frac{k \log(d/(k\delta))}{\epsilon^2}$ for some $C'>C$.

Using the previous results, we are now we are ready to give the main proof of \cref{thm_inf:LB_l2}, and hereby conclude the third and final part of our proof outline.

We construct an $\ell_2$-regression instance for which any $1 + \Theta(\eps)$ approximation derived from an oblivious sketch, paired with an arbitrary estimator, reveals a constant fraction of the support. The hardness result thus turns over to the regression problem.

\begin{proof}{(of \cref{thm_inf:LB_l2})}
	Let $X = \left[b\,\, A \right]$ be distributed as $G(I + zz^T)^{1/2}$, where $z$ is distributed as described in Lemma~\ref{lem:mutInfo}. 
	We will construct a $k$-sparse $\ell_2$ regression problem based on $X$ such that an $1+\epsilon$ approximation of the constructed problem allows us to recover a large enough fraction, i.e., greater than $19/20$, of the support of $z$. We can assume without loss of generality that the sketching matrix has orthonormal rows, and since $X$ is a Gaussian matrix, the sketch has rows sampled from the same distribution as the rows of $X$. This is so because we are proving lower bounds against any estimator on the sketch.   
 
    By Corollary~\ref{cor:fano}, if the number of samples is smaller than $ck \log(k/d) / \epsilon^2$, no algorithm can recover more than a $19/20$ fraction of the support of $z$ with probability larger than $2/3$. Therefore, we have a lower bound of $\Omega(k \log(d/k) / \epsilon^2)$ against a $1 + \epsilon$ approximation of $k$-sparse $\ell_2$ regression. 
	
	Consider the following $\ell_2$ sparse regression problem
	\[
	\min_{x \in \Psi_{k}}
	\left\lVert
	\left[
	\begin{array}{c}
		M \, M \, \ldots \, M \\
		A
	\end{array}
	\right] x 
	- 
	\left[
	\begin{array}{c}
		\sqrt{k} M \\
		b
	\end{array}
	\right]
	\right\rVert_2.
	\]
	We will let $M$ be a very large number, which enforces $\sum_i x(i)$ to be close to $\sqrt{k}$. The squared loss of the above regression problem is
	\[
	L = \| A {x} - b\|_2^2 + M^2 \left(\sum_i x(i) - \sqrt{k}\right)^2 = \|X \tilde{x}\|_2^2 + M^2 \left(\sum_i x(i) - \sqrt{k}\right)^2,
	\]
	where $\tilde{x} = (-1, {x})$. By matrix concentration in Lemma~\ref{lem:Gaussian_concentration}, for $n = \Omega(k\log(d/(k\delta)) /\epsilon^2)$, with probability at least $1 - \delta$,
	\[
	X^T X \approx_{\epsilon} n (I + zz^T).
	\]
	Recall that $z$ has $z(1) = 1$, $z(i) = \epsilon/\sqrt{k}$ for all $i \in U \setminus \{1\}$ and $z(i) = 0$ for all $i \notin U$. We have 
	
	\begin{align*}
		\frac{1}{n}\|X \tilde{x}\|_2^2 &\approx_{\eps} \| \tilde{x}\|_2^2 +  (\tilde{x}^T z)^2 
		= (1 + \|{x}\|_2^2) + (1 - \epsilon {x}^T {v})^2,
	\end{align*}
	where $v$ is such that $z = (1,v)$.
	
	Therefore, 
	\[
	\frac{L}{n} \approx_{\epsilon} 1 + \|{x}\|_2^2 + (1 - \epsilon {x}^T {v})^2 + \frac{M^2}{n} \left(\sum_i x(i) - \sqrt{k}\right)^2.
	\]
	Note that $x$ is a $1+\epsilon$ approximation of the $\ell_2$ loss iff  $x$ is an $1 + \Theta(\epsilon)$ approximation of the $\ell_2^2$ loss.
	
	Let  $\alpha = |\supp({x}) \cap \supp({v})|$. 
	By Lemma~\ref{lem:optimize_L}, there exists a constant $c$ such that if we can approximate $L$ within a factor of $1 + c \epsilon$, we must have $\alpha > 19k/20$. Rescaling $\epsilon$ and combining with Corollary~\ref{cor:fano} and a union bound (with $\delta$ set to a small constant) proves the theorem.
\end{proof}

Next, we extend our $\ell_2$ lower bound to $\ell_p$ for all $p\geq 1$.
\thmLBlp*
\begin{proof}{(of \cref{thm_inf:LB_lp})}
We reduce from the $\ell_2$ case by leveraging the fact that $\ell_2$ embeds obliviously up to $(1\pm\eps)$ distortion into $\ell_p$ for all $p\geq 1$ by Dvoretzky's theorem. Indeed, such an embedding can be constructed using a random mapping $G\in\R^{r \times n}$ whose entries are appropriately rescaled i.i.d. Gaussians. In particular $G$ is an oblivious linear map. The number of rows is $r=O(n\log(1/\eps)/\eps^2)$ for $1\leq p\leq 2$ and $r= n^{O(p)}$ for $p>2$; \cite[see][p. 30]{Matousek13}. We note, however, that the number of rows of $G$ does not matter in our context since it is reduced by an application of $S$ in what follows.

More precisely, we have for all $x\in\R^d$ that $(1-\eps) \|Ax-b\|_2 \leq \|GAx-Gb\|_p \leq (1+\eps) \|Ax-b\|_2$. Now suppose that $S$ is an oblivious sketching matrix for the $k$-sparse regression problem in $\ell_p$. Then we can find $\tilde x \in \argmin_{x\in\Psi_k} E_p(SGA,SGb,x)$. By definition it holds that $\|G(A\tilde x - b)\|_p\leq (1+\eps) \|G(Ax^*_G - b)\|_p$ where $x^*_G\in\argmin_{x\in\Psi_k}\|G(Ax-b)\|_p$. Also let $x^* \in \argmin_{x\in\Psi_k}\|Ax-b\|_2$ be the minimizer for the $\ell_2$ problem. Now it follows that
\begin{align*}
	\|A\tilde x - b\|_2 &\leq \|G(A\tilde x - b)\|_p / (1-\eps) \leq \|G(A x^*_G - b)\|_p (1+\eps) / (1-\eps) \\
	&\leq \|G(A x^* - b)\|_p (1+\eps) / (1-\eps) \leq \|A x^* - b\|_2 (1+\eps)^2 / (1-\eps) \\
	&\leq (1+7\eps) \|A x^* - b\|_2 ,
\end{align*}
which by rescaling $\eps$ means that $SG$ is an oblivious linear sketch for the $\ell_2$-norm problem with an $\ell_p$-norm minimization estimator and at most $1+\eps$ error. Using the $\ell_2$ lower bound given in \cref{thm_inf:LB_l2}, it follows that $SG$ and thus also $S$ has $m=\Omega(k \log(d/k) / \eps^2)$ rows.
\end{proof}

\subsection{Lower bounds for {ReLU} and hinge-like loss functions}
Here we give further reductions similar to \cref{thm_inf:LB_lp} in order to extend our main result to {ReLU} and hinge-like loss functions.

\thmLBrelu*
\begin{proof}
Note that $\|Ax-b\|_1 = \|Ax-b\|_{\relu} + \|-(Ax-b)\|_{\relu}, \; \forall x.$ Therefore, given $A, b$ we have $\|Ax-b\|_1=\|PA x-Pb\|_{\relu},$ where  
\[ P= \left[
\begin{array}{c}
	I_{n}\\
	-I_n 
\end{array}
\right]
\]
is a $2n \times n$ matrix.  Suppose $S\in\R^{m\times 2n}$ is an oblivious linear sketch for $k$-sparse $\relu$ regression with an estimator $E_{\relu}(SA,Sb,x)$.  This implies that  $\tilde x \in \argmin_{x\in\Psi_k} E_{\relu}(SPA,SPb,x)$ satisfies 
\[\|A\tilde x-b\|_1 = \|P(A\tilde x-b)\|_{\relu} \leq (1+\eps) \min_{x\in \Psi_k} \|P(Ax-b)\|_{\relu}=(1+\eps) \min_{x\in \Psi_k} \|Ax-b\|_1. 
\]
Therefore, $S\cdot P$ is an oblivious linear sketch for $\ell_1$ regression. By Theorem~\ref{thm_inf:LB_lp}, we have that $m=\Omega(k\log(d/k)/\eps^2).$
\end{proof}
\defhingelike*
The logistic loss function $\log (1 +e^{-x})$ and the hinge loss function $\max(0,1-x)$ are $(1,\ln(2), \ln(2))$ and $(1,1,1)$ hinge-like loss functions, respectively.

We will use the notation $a = (1\pm \epsilon)b$ to denote $(1-\epsilon)b \leq a \leq (1+\epsilon)b.$ We note that for $\epsilon>0$ if $\norm{x}_{\relu} >0,$ then there is a constant $c>0$ such that \[c\norm{x}_{\relu} = \norm{cx}_{\relu} = (1\pm \epsilon) \norm{cx}_{f}. \]

\thmLBhingelike*
\begin{proof}
Suppose $S\in\R^{m\times n}$ is an oblivious subspace embedding for $f$ with an estimator $E_f(SA,Sb,x)$. 

Let $x \in \Psi_k$ be some $k$-sparse vector. 
We will show that $S\in\R^{m\times n}$ is an oblivious subspace embedding for 
$\relu$ with the estimator 
$$ E_{\relu}(SA,Sb,x) = \lim_{c \rightarrow \infty} \frac{E_f(ScA,Scb,x)}{c}.$$

First, consider the case where $\|Ax-b\|_{\relu} > 0$. For $c$ large enough, we have 
\begin{align*}
	\|Ax-b\|_{\relu} &=  \frac{\|c(Ax-b)\|_{\relu}}{c} \\
	&= (1 \pm \epsilon) \frac{\|c(Ax-b)\|_f}{c}  \\
	&= (1 \pm 3\epsilon) \frac{E_f(ScA,Scb,x)}{c}. 
\end{align*}
Here, we assumed that $\epsilon$ is small enough so that $(1+\epsilon)^2 \leq 1 + 3 \epsilon$. 

Next we consider the case where $\|Ax-b\|_{\relu} = 0$. 

We have
\begin{align*}
	E_f(ScA,Scb,x) &\leq (1 + \epsilon) \|c(Ax-b)\|_f \\
	&\leq (1 + \epsilon) n a_1.
\end{align*}
Therefore, 
\begin{align*}
	\lim_{c \rightarrow \infty} \frac{E_f(ScA,Scb,x)}{c} &\leq  \lim_{c \rightarrow \infty} \frac{2 n a_1}{c}  = 0
\end{align*}
as desired.
\end{proof}

\subsection{Sampling fails for sparse regression}\label{sec:samp}

In this section we argue that sampling based algorithms, which are important tools in sketching for non-sparse regression \citep{DrineasMM06,DasguptaDHKM09,MunteanuSSW18,mai2021coresets}, do not give any non-trivial results in the sparse setting. Specifically, these algorithms cannot compress beyond what is possible in the non-sparse case -- roughly beyond the rank of the input matrix. This is the reason that our upper bounds all build on general linear sketches, rather than sampling.

\thmSamplingLower*
\begin{proof}
Let $k = 1$ and $A$ be the $n \times n$ identity matrix. Let $b$ be set to the $i^{th}$ standard basis vector with probability $1/n$. Note that $\displaystyle \min\nolimits_{1\text{-sparse } x \in \R^n} \norm{Ax-b} = 0$, so to achieve any bounded approximation factor, the algorithm must output $x$ with $0$ cost -- i.e., $x =b$. Any algorithm that accesses any $m < n/3$ rows of $A$ and $b$ will see only zero entries in $Sb$ with probability at least $2/3$. Let $\mathcal{X}_0$ be the distribution over outputs of the algorithm given that $Sb = 0$. The algorithm achieves a bounded approximation factor only if $x \sim \mathcal{X}_0$ satisfies $x = b$. This occurs with probability at most $\frac{3}{2n}$ since after seeing $Sb = 0$, any of the remaining $2/3 \cdot n$ possibilities for $b$ are equally likely. Thus, the algorithm succeeds with probability at most $1/3+3/(2n) < 1/2$ for $n > 9$.
\end{proof}
Note that the sampling matrix $S$ in Theorem \ref{thm_inf:samplingLower} may depend on $A$ but not on $b$. This is necessary. If the sampling matrix can depend on $b$, then, without bounded computation, a sampling algorithm can in theory compress the problem to $O(k/\epsilon)$ rows. In particular, it can simply solve for the optimal $\displaystyle x^* = \argmin_{k\text{-sparse } x \in \Psi_k} \norm{Ax-b}$ and only consider the $k$ columns of $A$ within the support of $x^*$. In e.g., the $\ell_2$ case, by simply applying variants of standard leverage score sampling, \citep{Woodruff14,chen2019active} to these columns, it can output a sampling matrix $S \in \R^{m \times n}$ for $m =O(k/\epsilon)$ with $\displaystyle\tilde x = \displaystyle\argmin_{k\text{-sparse } x \in \Psi_k} \norm{SAx-Sb}$ satisfying $\displaystyle \norm{A\tilde x - b} \le (1+\epsilon) \min\nolimits_{k\text{-sparse } x \in \Psi_k} \norm{Ax-b}$.

If the sampling matrix is required to preserve $\norm{Ax-b}$ for every $k$-sparse $x \in \Psi_k$, then even if a sampling algorithm can read $b$, it is easy to see that it must sample at least $n$ rows.  This is true even in the case that $b = 0$.
\thmSamplingLowerTwo*
\begin{proof}
Let $A$ be the $n\times n$ identity matrix. Then if we sample $m < n$ rows, we will have $\norm{SAx} = 0$ when $x$ is at least one of the standard basis vectors. This violates the approximation bound.
\end{proof}

\section{Upper bounds for \texorpdfstring{$k$}{k}-sparse regression}\label{sec:UB}
\subsection{Upper bounds for the \texorpdfstring{$\ell_p$}{ℓₚ}-norm loss function for \texorpdfstring{$p\in [1,2]$}{p ∈ [1,2]}}
We prove $k$-sparse affine embedding upper bounds and note that as a corollary we obtain the same bounds for minimization. We begin with $\ell_2$.
\thmUBltwo*
\begin{proof}{(of \cref{thm_inf:UB_l2})}
The upper bound is similar to the known constructions \citep{BaraniukDDW08} of RIP matrices via Johnson-Lindenstrauss embeddings \citep{JohnsL1984}, i.e., appropriately rescaled Gaussian matrices. The main difference is that the subspaces formed by any fixed $k$-sparse support of $x$ need not be orthogonal or aligned with the standard basis vectors. The first idea is that there are at most $\binom{d}{k} \leq (ed/k)^k$ different $k$-sparse supports and each of them corresponds to one choice of $k$ columns of $A$. Every such choice spans a $k$-dimensional linear subspace of dimension $\leq k$. By the subspace embedding construction in \citep{Sarlo2006}, every subspace formed by one choice of $k$ columns can be handled by embedding the points in a net of size $(3/\eps)^k$ covering the unit ball in the subspace. The remaining vectors can be related to the net points by triangle inequality and the embedding extends to vectors of arbitrary norm outside the unit sphere by linearity. A slightly more sophisticated argument in \cite[pp. 13]{Woodruff14} states that the net can be constructed with $\eps$ replaced by an absolute constant $\eps_0:=1/2$. So the total number of points to embed up to $(1\pm\eps)$ distortion is bounded by $|\calN|\leq (d/k)^k\cdot c^k$ for an absolute constant $c=3e/\eps_0$. The embedding can be accomplished via the Johnson-Lindenstrauss lemma followed by a union bound over $\calN$, which yields a matrix $S=\frac{1}{\sqrt{m}}G$ whose entries are scaled i.i.d. standard Gaussians $G_{ij}\sim N(0,1)$ with $m = O(\log(|\calN|)/\eps^2) = O(k\log(d/k)/\eps^2)$ rows.\footnote{We note that the same result can be achieved by random sign (Rademacher) matrices \citep{ClarkW09} which is more convenient in streaming and other space constrained settings.}
\end{proof}

We continue with $\ell_p, p\in [1,2)$. We note that the outline is similar to 
\cite[appendix F.1]{BackursIRW16} but our result is non-trivially adapted to the sparse setting and generalized to $\ell_p$.

\thmUBlonep*
\begin{proof}(of \cref{thm_inf:UB_l1_lp})
We choose $S \in \R^{m\times n}$ to be a matrix whose entries are i.i.d. $p$-stable random variables with scale parameter $\gamma_p = c^{1-1/p}$, where $c\approx 1.099055$; see below. We show that this matrix has the desired property. First note that $Ax-b = [A,b][x^T,-1]^T$, so we can simply assume that the input consist only of $A'=[A,b]$ and it suffices to show that $\|SA'x\|_{\rm med}=(1\pm\eps)\|A'x\|_p$ for all $x\in\Psi_{k'}\subseteq \R^{d+1}$ for $k'=k+1$. In what follows we re-substitute $A$ for $A'$ and $k$ for $k'$ for the sake of presentation.

Fix any $k$-sparse support indexed by $I\subseteq [d]$ with $|I|=k$. Let $A_I$ be the matrix whose $k$ columns are the columns $A_{*i}$ of $A$ such that $i\in I$. From a classic result of Auerbach \citep[cf.][]{Auerbach1930, DasguptaDHKM09} it follows that there exists a basis $L \in \R^{n\times k}$ for the $p$-normed subspace spanned by those columns \citep[][Lemma 2.22]{WangW22} that satisfies the following properties: the $\ell_p$ norm of each column $i\in[k]$ is exactly $\|L_{*i}\|_p = k^{1/q}$, and for all $x\in \R^k$ it holds that $\|Lx\|_p\geq\|x\|_{q}$, where $q=\infty$ for $p=1$ and $q=\frac{p}{p-1}$ for $p\in(1,2)$ denotes the \emph{dual} norm of $p$.

Further, by a property of the $p$-stable random variables we have that any entry $$(SL)_{ij}= \sum_{h=1}^n S_{ih}L_{hj} \sim C\cdot \|L_{*j}\|_p = C \cdot k^{1/q}$$ is again a scaled $p$-stable random variable; \citep[cf.][]{Indyk2006,Dytso18}. It follows that for any threshold $\tau$
the probability that any entry of $SL$ has absolute value larger than $\tau$ is bounded by $O((k^{1/q}/\tau)^p)=O(k/\tau^p)$, \citep[cf.][]{BednorzLM18}.

Setting $\tau = O((mk^2/\delta)^{1/p}) = \tilde{O}((k^3\log(d)/\delta)^{1/p})$, we have that all entries of $SL$ are simultaneously bounded by $\tau$ with probability $1-\delta/2$. Suppose this event is true. Then for all $x \in \Psi_k$ with the fixed support indexed by $I$, we have by Hoelder's inequality (applied only to the $k$-sparse support, since other terms are zero) and by the properties discussed above, that
\begin{align*}
	\normm{\infty}{SLx} &= \max_{i\in [m]} |S_{i*}Lx| \leq \max_{i\in [m]} \normm{\infty}{(SL)_{i*}} \normm{1}{x} \leq \tau k^{{1}-\frac{1}{q}} \normm{q}{x} \leq \tau {k}^{\frac{1}{p}} \normm{p}{Lx} \\
	&\leq \tilde{O}((k^{4}\log(d)/\delta)^{1/p}) \normm{p}{Lx}.
\end{align*}
Let $\tau' = \tilde{O}((k^{4}\log(d)/\delta)^{1/p}) = \tilde{O}(k^4\log(d)/\delta)$. We construct an $\frac{\eps}{\tau'}$-net $\calN_I^k$ in the $\ell_p$ norm for the unit $\ell_p$ ball intersect the subspace spanned by $L$. By linearity, the restriction to the unit ball is w.l.o.g. We repeat this construction for each $k$-sparse support and define our net to be the union over all supports, i.e., $\calN=\bigcup_{I\subseteq [d], |I|=k} \calN_I^k$. There are at most $\binom{d}{k}\leq (ed/k)^k$ different subspaces, each of which is covered by a net of size at most $|\calN_I^k|\leq (3\tau'/\eps)^k$ by the standard volume argument. Consequently for an absolute constant $c_1$ we have $|\calN|\leq (ed/k)^k\cdot (c_1 k^4\log(d)/(\eps\delta))^k = \exp( O( k\log(d/k) + k\log(k/(\eps\delta) ) ) )$.

Next, we investigate the cdf $F_p(x)$ of the random variable $|X|$, where $X$ follows a $p$-stable distribution. Except for the cases $p\in\{1,2\}$, which correspond to the Cauchy and Normal distribution, no analytic/closed form expression is known for the cdfs and pdfs. We thus take a detour and leverage the inversion theorem of L\'{e}vy \citep{levy1925,Masani1977} based on the characteristic function, for which a closed form expression is known \citep{Borak05,Dytso18}: $\phi_p(t)=\exp{-|\gamma_p t|^p}$, where $\gamma_p$ is the constant scale parameter defined above. More precisely, it holds that
\begin{align*}
	F_p(x)&=P(|X|\leq x)=P(X\leq x)-P(X\leq -x) \\
	&=\frac{1}{2\pi}\lim_{T\rightarrow \infty} \int^T_{-T} \frac{e^{itx} - e^{-itx}}{it}\phi_p(t)\,dt \\
	&=\frac{1}{2\pi}\lim_{T\rightarrow \infty} \int^T_{-T} \frac{e^{itx} - e^{-itx}}{it}e^{-|\gamma_p t|^p}\,dt \\
	&=\frac{1}{2\pi}\lim_{T\rightarrow \infty} \int^T_{-T} \frac{2i\sin(tx)}{it}e^{-|\gamma_p t|^p}\,dt \\
	&=\frac{1}{\pi}\lim_{T\rightarrow \infty} \int^T_{-T} \frac{\sin(tx)}{t}e^{-|\gamma_p t|^p}\,dt.
\end{align*}
It follows that
\begin{align}\label{eqn:median}
	F_p(1) &= \frac{1}{\pi}\lim_{T\rightarrow \infty} \int^T_{-T} \frac{\sin t}{t}e^{-|\gamma_p t|^p}\,dt  = \frac{1}{2}
\end{align}
It remains to show that the derivative of $F$ is bounded at $F(1)=\frac{1}{2}$. To this end we observe that 
\begin{align*}
	F_p'(x)&=\frac{1}{\pi}\lim_{T\rightarrow \infty} \int^T_{-T} \frac{\cos(tx)\cdot t}{t}\,e^{-|\gamma_p t|^p}\,dt \\
	&=\frac{1}{\pi}\lim_{T\rightarrow \infty} \int^T_{-T} {\cos(tx)}\,e^{-|\gamma_p t|^p}\,dt .
\end{align*}
Consequently,
\begin{align*}
	F_p'(1)&=\frac{1}{\pi}\lim_{T\rightarrow \infty} \int^T_{-T} {\cos(t)}\,e^{-|\gamma_p t|^p}\,dt . 
\end{align*}
Now by the symmetry of the integrand and monotonicity of the characteristic function w.r.t. the exponent $p$ we have 
\begin{align}\label{eqn:median_slope}
	\frac{1}{\pi} = \frac{1}{\pi}\lim_{T\rightarrow \infty} \int^T_{-T} {\cos(t)}\,e^{-|\gamma_1 t|}\,dt
	\leq F_p'(1)
	&\leq \frac{1}{\pi}\lim_{T\rightarrow \infty} \int^T_{-T} {\cos(t)}\,e^{-|\gamma_2 t|^2}\,dt \notag \\
	&\leq \frac{1}{\pi}\lim_{T\rightarrow \infty} \int^T_{-T} {\cos(t)}\,e^{-|t|^2}\,dt = \frac{1}{\pi^{1/2} e^{1/4}}
\end{align}
where in particular we note that $\gamma_1 = 1$, which is used in the lower bound, and $\gamma_2 \geq 1$ is used for the upper bound.

Generalizing \cite[Lemma 2]{Indyk2006}, it follows from \cref{eqn:median,eqn:median_slope} that if $F_p(z)\in [1/2-c\eps, 1/2+c\eps]$ for some absolute constant $c$ then $z\in [1-\eps, 1+\eps]$, which we will use in what follows. For any $x \in \Psi_k$, we say $SAx$ is \emph{good} if only a $\frac{1}{2}-c_2\eps$ fraction of coordinates in the sketch space are too large or too small, i.e.
\begin{align*}
	|\{i : |(SAx)_i| < (1-\eps) \normm{p}{Ax} \}| &\leq \left(\frac{1}{2} - c_2\eps \right)m\\
	|\{i : |(SAx)_i| > (1+\eps) \normm{p}{Ax} \}| &\leq \left(\frac{1}{2} - c_2\eps \right)m
\end{align*}
for some small constant $c_2$. If $SAx$ is \emph{good}, then for any $y$ with at most $c_2 \eps m$ coordinates larger than $\eps \normm{p}{Ax}$, we have
\begin{align}
	(1-2\eps) \normm{p}{Ax} \leq \normm{\rm med}{SAx + y} \leq
	(1+2\eps)\normm{p}{Ax}.\label{eq:medadjust}
\end{align}    

By the $p$-stability property, $(SAx)_i$ is a $p$-stable random variable with scale $\normm{p}{Ax}$, we
have that
\begin{align*}
	\Pr[|(SAx)_i| < (1-\eps)\normm{p}{Ax}] &< 1/2 - \Omega(\eps)\\
	\Pr[|(SAx)_i| > (1+\eps)\normm{p}{Ax}] &< 1/2 - \Omega(\eps).
\end{align*}
By a Chernoff bound, for sufficiently small $c_2$ we have that $SAx$ is \emph{good} with probability at least $1-\exp(-\Omega(\eps^2 m))$. For our choice of $m$, we can union bound to get that $SAx$ is \emph{good} simultaneously for all $x \in \calN$ with probability at least $1-\exp({-\Omega(\eps^2 m)})\cdot|\mathcal{N}| \geq 1-\delta^{\Omega(k)}$. Suppose this event is true.

Then every $y=Ax$ for $x\in \Psi_k$ with $\normm{p}{y} = 1$ can be expressed as $y=z + \eta$ where $z \in \calN$ and $\normm{p}{\eta} \leq \eps/\tau'$. We have that $Sz$ is \emph{good} and that $\normm{\infty}{S\eta} \leq \tau' \normm{p}{\eta} \leq \eps$. Hence by~\eqref{eq:medadjust},
\[
(1-2\eps) \normm{p}{z} \leq \normm{\rm med}{S(z + \eta)} \leq (1+2\eps)\normm{p}{z}.
\]
which implies
\[
(1-3\eps) \normm{p}{y} \leq \normm{\rm med}{Sy} \leq (1+3\eps)\normm{p}{y}.
\]
Since $S$ is linear, the restriction to $\normm{p}{y} = 1$ is not necessary. Rescaling $\eps$ concludes the proof.
\end{proof}

\subsection{Upper bounds for the {ReLU} loss function}

\paragraph{Notation} For a function $f: \R \rightarrow \R$ and a vector $y \in \mathbb{R}^n$, we let $f(y) \in \mathbb{R}^n$ denote the entry-wise application of $f$ to $y$. Let $y_i$ denote the $i^{th}$ entry of $y$. So $f(y)_i = f(y_i)$. Moreover, let $\norm{y}_f = \sum_{i=1}^n f(y_i)$.
Let $y^+$ and $y^-$ denote $y$ restricted to the set of positive and negative entries respectively. Finally, for $A \in \mathbb{R}^{n \times d}$ and $x \in \mathbb{R}^d$, let $\mu(A) = \sup_{x \ne 0} \frac{(Ax)^+}{(Ax)^-}$. When $A$ is clear from the context, we drop $A$ from the notation of $\mu$.

\thmUBrelufromlone*
\begin{proof}
By appending $b$ to $A$ and increasing $k$ by 1, if suffices to prove the statement of the theorem for the case $b=0$. Note that
\begin{align*} 
	\norm{Ax}_{\relu} = \norm{(Ax)^+}_1 =  
	\frac{\norm{Ax}_1 + \mathbf{1}^TAx}{2} 
\end{align*}
since $\norm{Ax}_1 = \norm{(Ax)^+}_1 + \norm{(Ax)^{-}}_1$ and $\mathbf{1}^TAx = \norm{(Ax)^+}_1 - \norm{(Ax)^-}_1$.
From \cref{thm_inf:UB_l1_lp}, there exists a sketch $S_{\ell_1}$ with $O\left( ({ k}/{\epsilon^2}) \log ({d}/{\epsilon \delta})\right)$ rows and an estimator $g_{\ell_1}(S_{\ell_1}A, x) = \norm{S_{\ell_1} A x}_{{\rm med}}$ such that, with probability at least $1-\delta$, 
\begin{equation} 
	\label{eqn:l1_approximation}
	\forall x\in\Psi_k: (1-\epsilon) \norm{Ax}_{1}\leq g_{\ell_1}(S_{\ell_1}A, x) \leq (1 + \epsilon) \norm{Ax}_{1}. 
\end{equation}
Moreover, $\mathbf{1}^TA$ can be computed exactly using a single row. Let $S$ be the sketch obtained by combining $S_{\ell_1}$ and the single row $\mathbf{1}^T$, and let
$$g_{\relu}(SA,x) \eqdef  \frac{g_{\ell_1}(S_{\ell_1}A,x)+\mathbf{1}^T Ax}{2} = \frac{\norm{S_{\ell_1} A x}_{{\rm med}}+\mathbf{1}^T Ax}{2}.$$
Now, \eqref{eqn:l1_approximation} implies
\begin{align}
	\label{eqn:relu_plus_epsilon_l1}
	\big| g_{\relu}(SA,x) - \norm{Ax}_{\relu} \big| 
	&= \left| \frac{g_{\ell_1}(S_{\ell_1}A,x)+\mathbf{1}^T Ax}{2} - \frac{\norm{Ax}_1 + \mathbf{1}^TAx}{2} \right| \nonumber \\
	&= \left| \frac{ g_{\ell_1}(S_{\ell_1}A,x) - \norm{Ax}_1 }{2} \right| \nonumber \\
	&\leq \frac{\epsilon}{2} \norm{Ax}_1.
\end{align}
By the definition of $\mu$,
\begin{align}
	\label{eqn:relate_relu_l1_by_mu}
	\mu + 1 \geq \frac{\norm{(Ax)^-}_1}{\norm{(Ax)^+}_1} + 1 = \frac{\norm{Ax}_1}{\norm{Ax}_{\relu}}.
\end{align}
From \eqref{eqn:relu_plus_epsilon_l1} and \eqref{eqn:relate_relu_l1_by_mu}, we have
\begin{align*}
	\big| g_{\relu}(SA,x) - \norm{Ax}_{\relu} \big| &\leq  \frac{\epsilon(\mu+1)}{2}\norm{Ax}_{\relu} \leq \epsilon \mu \norm{Ax}_{\relu},
\end{align*}
where the last inequality holds because $\mu \geq 1$.
The theorem follows by scaling $\epsilon$ by a factor of $1/\mu$.
\end{proof}

\subsection{Upper bounds for hinge-like loss functions}
First we give a lemma on hinge-like functions, which has a similar role to \eqref{eqn:relate_relu_l1_by_mu} for the $\relu$ function, and was proven in Corollary 9 of \cite{mai2021coresets}. 
\begin{lemma}
\label{lem:lowerbound_f_norm_by_epsilon_1_norm_plus_n}
Let $f$ be an $(L, a_1, a_2)$ hinge-like loss function, and let $C = 16 \max(1,L,a_1,1/a_2)^4$. Let $A\in \R^{n\times d}$. Then for any $x \in \R^d$,
\begin{align*}
	\norm{Ax}_f \geq \frac{n + \norm{Ax}_1}{C \mu}.
\end{align*}
\end{lemma}
\begin{proof}
We have
\begin{align}\label{eq:l1Bound}
	\norm{Ax}_f 
	&\ge \sum_{i: [Ax]_i \in [0,2a_1]} f(Ax)_i +  \sum_{i: [Ax]_i \ge 2a_1} f(Ax)_i\nonumber\\
	&\ge \sum_{i: [Ax]_i \in [0,2a_1]} a_2 +  \sum_{i: [Ax]_i \ge 2a_1} \relu(Ax)_i - a_1 \nonumber \\
	&\ge \min \left (\frac{a_2}{2a_1}, \frac{1}{2} \right )\cdot \norm{(Ax)^+}_1\nonumber\\
	&\ge \min \left (\frac{a_2}{2a_1}, \frac{1}{2} \right )\cdot \frac{\norm{Ax}_1}{\mu+1}, 
\end{align}
where the second inequality holds because $f$ is $(L, a_1, a_2)$ hinge-like and last inequality follows from \eqref{eqn:relate_relu_l1_by_mu} in \cref{thm_inf:UB_relu_from_l1}.

Let $\gamma \eqdef \min \left (\frac{a_2}{2a_1}, \frac{1}{2} \right )$. 
Now we claim that 
$$\norm{Ax}_f = \sum_{i=1}^n f(Ax)_i \ge \frac{n a_2 \gamma}{4 \mu \cdot \max(1,L)}.$$ 
If $\sum_{i=1}^n f(Ax)_i \ge \frac{n a_2}{4}$ then this holds immediately since $\mu(X) \ge 1$, $\max(1,L) \ge 1$ and $\gamma \le 1$. Otherwise, assume that $\sum_{i=1}^n f(Ax)_i \le \frac{n a_2}{4}.$ Since $f(z) \ge a_2$ for all $z \ge 0$ and since $f$ is $L$-Lipschitz, $f \left (z \right ) \ge \frac{a_2}{2}$ for all $z \ge -\frac{a_2}{2L}$. This implies that $Ax$ has at most  $\frac{na_2/4}{a_2/2} = \frac{n}{2}$ entries $\ge  -\frac{a_2}{2L}$. Thus, $Ax$ has at least $\frac{n}{2}$ entries $\le -\frac{a_2}{2L}$ and so  $\norm{(Ax)^-}_1 \ge \frac{n a_2}{4 L}$. Thus, by the definition of $\mu$ along with \eqref{eq:l1Bound},
\begin{align}\label{eq:cBound}
	\norm{Ax}_f \ge \gamma \cdot \norm{(Ax)^+}_1 \ge \frac{n a_2 \gamma}{4 \mu L } \ge \frac{n a_2 \gamma}{4 \mu \cdot \max(1,L)}.
\end{align}
Combining \eqref{eq:l1Bound} with \eqref{eq:cBound} gives that
\begin{align*}
	\norm{Ax}_f &\ge \frac{\gamma \cdot \norm{Ax}_1}{2\mu +2} + \frac{n a_2 \gamma}{8 \mu \cdot \max(1,L) } \\
	&\ge \left (\norm{Ax}_1 + n \right ) \cdot  \frac{\gamma \cdot \min(1,a_2)}{8 \mu \cdot \max(1,L)}
	\\
	&\ge \left (\norm{Ax}_1 + n \right ) \frac{1}{8 \mu \cdot \max(1,L) \cdot \max(1, 1/a_2) \cdot \max(2,2a_1/a_2)} \\
	&\ge \left (\norm{Ax}_1 + n \right ) \frac{1}{16 \mu \cdot \max(1,L,a_1,1/a_2)^4}. 
\end{align*}
Substituting $C = 16 \max(1,L,a_1,1/a_2)^4$ completes the proof. 
\end{proof}

From \cref{lem:lowerbound_f_norm_by_epsilon_1_norm_plus_n}, it suffices to approximate $\norm{Ax}_f$ within $O((\epsilon/\mu)(n + \norm{Ax}_1))$ to obtain a relative error guarantee. \cref{thm_inf:UB_relu_from_l1} provides a method to approximate $\norm{Ax}_{\relu}$ within $O((\epsilon/\mu)\norm{Ax}_1)$. In \cref{thm_inf:UB_hingelike_from_relu}, we will show that uniform sampling can approximate the difference between $\norm{Ax}_{f}$ and $\norm{Ax}_{\relu}$ within $O((\epsilon /\mu)n)$.

\thmUBhingelikefromrelu*
\begin{proof}
Again, we may assume that $b=0$. We have 
\begin{align*}
	\norm{Ax}_f = \sum_{i=1}^n f(Ax)_i = \sum_{i=1}^n \left( f(Ax)_i - \relu(Ax)_i \right) + \norm{Ax}_{\relu}.
\end{align*}
By \eqref{eqn:relu_plus_epsilon_l1} in the proof of Theorem~\ref{thm_inf:UB_relu_from_l1}, there exists a sketch $S_{\relu}$ with $m_1 = O\left({k}\log({d}/{\epsilon \delta})/{\epsilon^2} \right)$ rows and a function $g_{\relu}$ such that with probability at least $1 - \delta$,
\begin{align}
	\label{eqn:relu_l1_mu_approx}
	\big| g_{\relu}(S_{\relu}A,x) - \norm{Ax}_{\relu}\big| \leq \epsilon \norm{Ax}_1.
\end{align}

We give a sketch with $m_2$ rows to approximate $R(Ax) = \sum_{i=1}^n \left( f(Ax)_i - \relu(Ax)_i \right)$. Consider uniformly sampling the rows of $A$ with replacement. Let $S_u \in \mathbb{R}^{m_2 \times n}$ be the sketching matrix corresponding to uniformly sampling $m_2$ rows.
We will show that 
$$
R(S_u A x) = \frac{n}{m_2} \sum_{i=1}^{m_2} \left( f(S_u Ax)_i - \relu(S_u Ax)_i \right)
$$
can approximate $R(Ax)$ within error $O(\epsilon(n + \norm{Ax}_1))$ for a suitable value of $m_2$, i.e.,
\begin{align}
	\label{eqn:approx_R}
	\left| R(S_u A x)  - R(Ax) \right| \leq  O(\epsilon(n + \norm{Ax}_1)).
\end{align}

Fix any $k$-sparse support indexed by $I\subseteq [d]$ with $|I|=k$. We will show \eqref{eqn:approx_R}  for all $x$ in this fixed support and then union bound over all $\binom{d}{k}$ supports. Let $A_I$ be the matrix whose $k$ columns are the columns $A_{*i}$ of $A$ such that $i\in I$. We may assume w.l.o.g. that the columns of $A_I$ are orthonormal, since the set of vectors $\{A_I x \mid x \text{ has support } I \}$ remains unchanged by making them orthonormal. With the assumption, $\norm{A x}_2 = \norm{A_I x}_2 = \norm{x}_2$ for all $x$ having support $I$. We consider two cases.

\paragraph{Large Norm} In this case, we consider $x$ such that $\norm{Ax}_1 \geq n a_1/\epsilon$. Since $f$ is an $(L,a_1, a_2)$ hinge-like loss function, $|f(Ax)_i - \relu(Ax)_i| \leq a_1$ for all $1 \leq i \leq n$. Therefore, it holds that $\left| f(S_uAx)_i - \relu(S_uAx)_i \right| \leq a_1$ for all $1 \leq i \leq m_2$ as well. We have 
\begin{align}
	\label{eqn:large_norm}
	\big| R(S_u A x) - R(Ax) \big| &= \left| \frac{n}{m_2}\sum_{i=1}^{m_2} \left( f(S_u A x)_i - \relu(S_u A x)_i \right)  - \sum_{i=1}^n \left( f(Ax)_i - \relu(Ax)_i \right) \right| \nonumber \\
	&\leq \frac{n}{m_2} m_2 a_1  + n a_1 = 2na_1 \leq 2\epsilon \norm{Ax}_1.
\end{align}
\paragraph{Small Norm} Now we consider $x$ such that $\norm{Ax}_1 < n a_1/\epsilon$. This implies
$\norm{x}_2 = \norm{Ax}_2 \leq \norm{Ax}_1 < n a_1/\epsilon$.
We construct an 
$\frac{\epsilon}{L+1}$-net $\mathcal{N}_I^k$ in the $\ell_2$ norm for all $x$ in the $\ell_2$ ball 
$\mathcal{B}_I = \{x \text{ has support } I \mid \norm{x}_2 \leq n a_1/\epsilon \}$. By a standard volume argument, $\left| \mathcal{N}_I^k \right| < ({3  (L+1) n a_1}/{\epsilon^2})^k$. We next consider a fixed vector $x$ with support $I$ and then union bound over all vectors in the net $\mathcal{N}_I^k$.

Let $X_1, X_2, \ldots, X_{m_2}$ be random variables such that $X_i = f(S_u A x)_i - \relu(S_u Ax)_i$. Since $f$ is an $(L,a_1, a_2)$ hinge-like loss function, $|X_i| \leq a_1$ for all $i$. By Hoeffding's inequality, 
\begin{align}
	\label{eqn:hoeffding_uniform}
	\pr\left( \left| \frac{n\sum_{i=1}^{m_2} X_i}{m_2} - R(Ax) \right| \geq  {\epsilon n} \right) = 
	\pr\left(  \left| \frac{\sum_{i=1}^{m_2} X_i}{m_2} - \frac{R(Ax)}{n}  \right|  \geq  {\epsilon} \right) \leq 2 \exp{\left(-\frac{m_2 \epsilon^2}{4 a_1^2}\right)}.
\end{align}
Note that \eqref{eqn:hoeffding_uniform} holds for a fixed vector $x$. Letting
\begin{align*}
	m_2 = O\left( \frac{a_1^2}{\epsilon^2} \left( k \log \frac{d (L+1)n a_1}{\epsilon k}  + \log \frac{1}{\delta} \right) \right)
\end{align*} 
and union bounding over at most $({3  (L+1) n a_1}/{\epsilon^2})^k$ points in $\mathcal{N}_I^k$, we have the probability that \eqref{eqn:hoeffding_uniform} holds for all points in $\mathcal{N}_I^k$ is 
$1 - \delta' $, where 
\begin{align*}
	\delta' = \left( \frac{3(L+1)n a_1}{\epsilon^2 } \right)^k \exp{\left(-\frac{m_2 \epsilon^2}{4 a_1^2}\right)} \leq \delta \left( \frac{k}{ed} \right)^k.
\end{align*}

Next we will show that,
$
\big| R(S_u A x') - R(A x') \big| \leq \epsilon n 
$
for all $x'$ in $\mathcal{N}_I^k$ implies 
\begin{align}
	\label{eqn:uniform_n_mu_approx}
	\big| R(S_u A x) - R(A x) \big| \leq 3 \epsilon n
\end{align}
for all $x \in B_I$.
Let $x' \in \mathcal{N}_I^k$ such that $\norm{x-x'}_2 \leq \frac{\epsilon}{L+1}$. 
By the triangle inequality, 
\begin{align*}
	\big| R(S_u A x) - R(Ax) \big| & \leq  \big| R(A x) - R(A x') \big| + \big| R(S_u A x) - R(S_u A x') \big| + \big| R(S_u A x') - R(Ax') \big| .
\end{align*}
Since $x'$ is in $\mathcal{N}$, the last term is at most $\epsilon n$. We will bound the other two terms using the fact that $f$ is $L$-Lipschitz and $\norm{x-x'}_2 \leq \frac{\epsilon}{ L+1}$. Applying the triangle inequality again, we have 
\begin{align*}
	\big| R(Ax) - R(Ax') \big| & \leq  \left| \sum_{i=1}^n \Big( \relu(Ax)_i - \relu(Ax')_i \Big) \right| + \left| \sum_{i=1}^n \Big( f(Ax)_i - f(Ax')_i \Big) \right| \\
	& \leq \norm{A(x-x')}_1 + L \norm{A(x-x')}_1 \\
	& \leq (1+L) \sqrt{n} \norm{x-x'}_2 \\
	& \leq \epsilon \sqrt{n} 
	.\end{align*}
Similarly, 
\begin{align*}
	~&~ \big| R(S_u A x) - R(S_u A x') \big| \\ & \leq  \left| \frac{n}{m_2}\sum_{i=1}^{m_2} \Big(  \relu(S_u A x)_i - \relu(S_u A x')_i \Big)  \right| + \left| \frac{n}{m_2} \sum_{i=1}^{m_2} \Big(  f(S_u A x)_i - f(S_u A x')_i \Big) \right| \\
	& \leq \frac{n}{m_2}  \norm{S_u A (x-x')}_{1} + \frac{n}{m_2}  L \norm{S_u A (x-x')}_{1} \\
	& \leq (L+1) \frac{n}{\sqrt{m_2}} \norm{S_u A(x-x')}_{2} \\
	& \leq (L+1) \frac{n}{\sqrt{m_2}} \norm{A(x-x')}_{2} \\
	& \leq (L+1) \frac{n}{\sqrt{m_2}} \norm{x-x'}_2 \\
	& \leq \epsilon  \frac{n}{\sqrt{m_2}} .
\end{align*}
This completes the proof of \eqref{eqn:uniform_n_mu_approx} and our argument for the small norm case. Combining \eqref{eqn:large_norm} and \eqref{eqn:uniform_n_mu_approx}, we have that for a fixed support $I$, with probability at least $1 - \delta \left( \frac{k}{ed} \right)^k$,
\begin{align*}
	\big| R(S_u A x) - R(A x) \big| \leq 3\epsilon n + 2 \epsilon \norm{Ax}_1.
\end{align*}
Union bounding over $\binom{d}{k}<(ed/k)^k$ $k$-sparse supports gives a success probability of at least $1-\delta$.

Let $S= \left[
\begin{array}{c}
	S_{\relu} \\
	S_u 
\end{array}
\right]$. Define 
\[
g_f(SA,x) \eqdef R(S_u A x) + g_{\relu}(S_{\relu}A,x).
\]
From \eqref{eqn:relu_l1_mu_approx}, with probability at least $1-2\delta$, 
\begin{align}
	\label{eqn:hinge_combined_1}
	\big| g_f(SA,x) - \norm{Ax}_f\big| & \leq \big| R(S_u A x) - R(A x) \big| + \big| g_{\relu}(S_{\relu}A,x) - \norm{Ax}_{\relu}\big| \nonumber \\
	& \leq 3\epsilon \left( n + \norm{Ax}_1\right)
\end{align}
By \cref{lem:lowerbound_f_norm_by_epsilon_1_norm_plus_n},
$
\norm{Ax}_f \geq {(n + \norm{Ax}_1})/{(C \mu)},
$
where $C= 16 \max(1,L,a_1,1/a_2)^4$. Combining this with \eqref{eqn:hinge_combined_1} gives
\begin{align*}
	\big| g_f(SA,x) - \norm{Ax}_f\big| \leq {3C\mu\epsilon}  \norm{Ax}_f.
\end{align*}
Finally, scaling $\epsilon$ by $O({1}/({C \mu}))$, the number of rows in $S$ is
\begin{align*}
	m_1 + m_2 &= O\left( \frac{C^2\mu^2k(1+a_1^2)}{\epsilon^2} \cdot \log \left( \frac{Cna_1(L+1)\mu d}{\epsilon \delta} \right) \right) \\
	&= O\left( \frac{c^{10}\mu^2k}{\epsilon^2} \cdot \log \left( \frac{cn\mu d}{\epsilon \delta} \right) \right)
\end{align*}
with $c = \max(1,L,a_1,1/a_2)$.
\end{proof}

For completeness we have the following simple result that yields a connection between minimizing in the sketch space using a $k$-sparse affine embedding -- as in all upper bounds above -- and the original problem.

\corfminimization*
\begin{proof}(of \cref{cor_inf:f_minimization})
Let ${x}^* \in \argmin_{x \in \Psi_k} \norm{Ax-b}$. Then
\begin{align*}
	\|A\tilde{x}-b\| &\leq E(SA,Sb,\tilde{x})/(1-\eps) \leq E(SA,Sb,{x}^*)/(1-\eps) \\
	&\leq \norm{Ax^*-b}(1-\eps)/(1-\eps) \leq (1+4\eps) \norm{Ax^*-b}.
\end{align*}
\end{proof}

\section{Oblivious sketching for LASSO regression}

LASSO regression is a convex relaxation of $k$-sparse $\ell_2$ regression and enjoys large popularity as a heuristic for inducing sparsity and feature selection \citep{Tibshirani96}. Here we give an upper bound for sketching that depends on an $\ell_1$ regularization parameter $\lambda$, and $\log(d)$.

\subsection{Upper bound}

\thmLASSOUpper*
Observe that our constraints on $\norm{A}_2,\norm{b}_2$ are necessary, since classic LASSO regression is \emph{not scale invariant}. If, for an arbitrarily large factor $\alpha$, we scale $A$ up by a factor of $\alpha$ then we can keep the same error $\norm{Ax-b}_2^2$ while scaling $x$ down by a factor of $1/\alpha$. Thus, the $\lambda \cdot \norm{x}_1$ term becomes negligible, and the problem reduces to ordinary least squares regression. Similarly, if we scale $b$ up by a factor $\alpha$, then the minimal achievable $\norm{Ax-b}_2^2$ scales by a factor $\alpha^2$, while the $x$ achieving this minimum has $\lambda \norm{x}_1$ scaled by just a factor $\alpha$. So again, as $\alpha$ grows arbitrarily large, the problem becomes ordinary least squares regression. It is known that sketching dimension $\Theta(d/\eps)$ is necessary and sufficient for ordinary least squares regression \citep{Sarlo2006,ClarkW09}. Thus, going beyond this requires bounding $\norm{A}_2,\norm{b}_2$ to ensure that the regularization $\lambda \norm{x}_1$ has a non-negligible effect.

Also note that we cannot hope to achieve an $o(d/\epsilon^2)$ bound for preserving the LASSO cost for all $x \in \R^d$, since for $x$ with large enough $\norm{x}_2^2$, and for $b = 0$, the problem becomes equivalent to preserving $\norm{Ax}_2^2$, which requires $\Theta(d/\epsilon^2)$ sketch size \citep{NelsonN14}.

\begin{proof}[Proof of Theorem \ref{thm_inf:LASSOUpper}]
Let $OPT = \min_{x \in \R^d} \norm{Ax-b}_2^2 + \lambda \norm{x}_1$. Observe that since $\norm{b}_2 \le 1$, $OPT \le \norm{b}_2^2 + \lambda \norm{0}_1 \le 1$.
Via the standard Johnson-Lindenstrauss lemma, we have that with high probability, for $x^* = \argmin_{x \in \R^d} \norm{Ax-b}_2^2 + \lambda \norm{x}_1$, $\norm{SAx^*-  Sb}_2^2 + \lambda \norm{x^*}_1 \le (1+\epsilon) \cdot OPT$.

Thus, we must have $\norm{\tilde x}_1 \le \frac{(1+\epsilon) \cdot OPT}{\lambda} \le \frac{2 \cdot OPT}{\lambda }$, as otherwise we would have $\norm{SA\tilde x-  Sb}_2^2 + \lambda \norm{\tilde x}_1 \ge\norm{SAx^* - Sb}_2^2 + \lambda \norm{x^*}_1$, contradicting the fact that $\tilde x$ is a minimizer for the sketched problem. For the same reason, we must have $\norm{SA\tilde x -Sb}_2 \le \norm{Sb}_2 \le 2.$ 

Let $\mathcal{T} = \{y = Ax-b: \norm{x}_1 \le \frac{2 \cdot OPT}{\lambda }\}$. Let $\mathcal{T}' = \{y = Ax-b: \norm{x}_1 = 1\}$. Observe that by our assumption that $\norm{A}_2\le 1$, each column $a_i$ of $A$ has $\norm{a_i}_2 \le 1$. Thus, $\mathcal{T}'$ is the convex hull of $2d + 1$ points in the unit ball: $a_1,-a_1,a_2,-a_2,\ldots a_d,-a_d,b$. By Corollary 3.2 of \cite{NaN18}, for $m = O \left (\frac{\log(d/\delta)}{\lambda^2 \cdot \epsilon^2}\right )$, with probability at least $1-\delta$, for all $y' \in \mathcal{T}'$,
\begin{align*}
	\left | \norm{Sy}_2 - \norm{y}_2 \right | \le \epsilon \lambda.
\end{align*}
Note that any $y \in \mathcal{T}$ can be written as $\alpha \cdot y'$ for $y' \in \mathcal{T'}$, for some $\alpha \le \frac{2 \cdot OPT}{\lambda}$. Thus, we have that with probability at least $1-\delta$, for all $y \in \mathcal{T}$,
\begin{align*}
	\left | \norm{Sy}_2 - \norm{y}_2 \right | \le 2\epsilon \cdot OPT.
\end{align*}

In particular, for $\tilde x = \argmin_{x \in \R^d} \norm{SAx-Sb}_2^2  + \lambda \norm{x}_1$, we have $A\tilde x-b \in \mathcal{T}$ so this gives
\begin{align}\label{lassAdd}
	\left | \norm{SAx-Sb}_2^2 - \norm{Ax-b}_2^2 \right | &\le \left | \norm{SAx-Sb}_2 - \norm{Ax-b}_2 \right | \cdot \left | \norm{SAx-Sb}_2 + \norm{Ax-b}_2 \right |\nonumber\\
	&\le 12\epsilon \cdot OPT,
\end{align}
where the second inequality uses that $\norm{SA \tilde x - Sb}_2 \le 2\cdot OPT \le 2$ and thus $\norm{A \tilde x - b}_2 \le (2+2\epsilon)\cdot OPT \le 4.$ Finally, using \eqref{lassAdd} we have:
\begin{align*}
	\norm{A\tilde x-b}_2^2 + \lambda \norm{\tilde x}_1 &\le \norm{SA\tilde x-Sb}_2^2 + \lambda \norm{\tilde x}_1 + 5 \epsilon \cdot OPT\\
	&\le \norm{SA\tilde x^*-Sb}_2^2 + \lambda \norm{x^*}_1 + 12 \epsilon \cdot OPT\\
	&\le (1+13\epsilon) \cdot [\norm{SA\tilde x^*-Sb}_2^2 + \lambda \norm{x^*}_1].
\end{align*}
This completes the theorem after adjusting $\epsilon$ by a constant factor.
\end{proof}

\subsection{Lower bound}
\thmLASSOLower*
\begin{proof}[Proof of \cref{thm_inf:LASSOLower}]
We will prove this similar to the lower bound in Theorem~\ref{thm_inf:LB_l2}. We will take $[b \, \, A] = X \sim \frac{1}{\sqrt{n}}G(I + zz^T)^{1/2},$ where $z$ is as in Corollary~\ref{cor:fano} with  $\epsilon = 1/2.$ Without loss of generality, we can assume that $S$ has orthonormal rows and so $SG(I + zz^T)^{1/2}$  has the same distribution as $G(I + zz^T)^{1/2}$ with fewer rows. We also let $\lambda = 1/(2\sqrt{k}).$ 
The normalization factor $\frac{1}{\sqrt{n}}$ ensures that norm condition $\|A\|_2, \|b\|_2 \leq 1$ holds with high probability. 
We then have with high probability
\[
\norm{A x - b}_2^2 + \lambda \norm{ x}_1 \approx 1 + \|x\|_2^2 + (1 - \epsilon v^T x)^2 + \lambda \|x\|_1 =: L(x).
\]
The approximation error (and the probability) above can be made arbitrarily small by taking $n$ to be sufficiently large.

We note that $L(x)$ is a $1$-strongly convex function, and so  \[
L(\hat{x}) \geq L(x^*) + \|\hat{x} - x^*\|^2_2.
\]
for any $\hat{x}$, where $x^*$ is the minimizer of $L(x).$
By a straightforward computation, we also get that $x^* = v/5$ and $L(x^*) = 1.95.$ Suppose that $L(\hat{x}) \leq (1+c_1) L(x^*)$ for a sufficiently small $c_1.$ Then we have 
\[
(1+c_1) L(x^*) \geq L(x^*) + \|\hat{x} - x^*\|^2_2.
\]
This implies that $\|\hat{x} - x^*\|^2_2 \leq c_1 L(x^*) \leq 1.95 \cdot c_1.  $ Therefore, by choosing $c_1$ to be a small enough constant, we can recover a $19/20$ fraction of  $\mbox{supp}(x^*)=\mbox{supp}(v).$ Corollary~\ref{cor:fano} then gives the required lower bound of $\Omega(k \log(d/k))$ on the size of the sketch, where $k$ is set to $1/(4\lambda^2)$ corresponding to our choice of $\lambda$.
\end{proof}

\section{Separating sparse recovery and sparse regression}
We now give a separation between the $k$-sparse recovery problem and $k$-sparse regression problems. Combined with our lower bounds for $k$-sparse regression, this shows that the $k$-sparse recovery problem is a strictly easier problem. Our upper bound below matches the lower bound for $k$-sparse outputs of \cite{price20111+} up to a $\log(k/\epsilon)$ factor, improving the na\"ive bound of $O(k (\log d) / \epsilon^2)$. 

In the $k$-sparse recovery problem, one seeks to sketch a vector $x \in \mathbb{R}^d$ so as to output a $k$-sparse $\hat{x} \in \mathbb{R}^d$ so
that 
\begin{eqnarray}\label{eqn:prop}
\|x - \hat{x}\|_2^2 \leq (1+\epsilon) \|x-x_k\|_2^2,
\end{eqnarray}
where $x_k$ consists of the top $k$ entries in magnitude of $x$, breaking ties arbitrarily. We
need the following theorem about CountSketch.

\begin{theorem}(\cite{CharikarCF04})\label{thm:countsketch}
There is a distribution on sketching matrices $S \in \mathbb{R}^{bt \times d}$, called Count\-Sketch matrices,
which is parameterized by the number $b$ of buckets and the number $t$ of tables. For a vector $x$, there is a procedure which, given $S \cdot x$ and a coordinate $i \in [d]$,  
returns an estimate $\hat{x}_i$ for which
$$|\hat{x}_i - x_i| \leq C \cdot \|x-x_k\|_2/\sqrt{b},$$
with failure probability at most $2^{-C't}$, 
where $C, C' > 0$ are absolute constants. 
\end{theorem}

\noindent
Consider the following procedure:
\begin{enumerate}
\item Run CountSketch with $b = O(k/\epsilon)$ buckets and $t = O(\log d)$ tables, and let $A$ be the set of indices $i \in \{1, 2, \ldots, d\}$ for which the corresponding estimates $\hat{x}_i'$ are among the largest $O(k/\epsilon)$ in magnitude. Here 
we use $\hat{x}_i'$ to denote the estimates returned by Theorem \ref{thm:countsketch} to distinguish them from the estimates returned in the next step. 
\item In parallel, run CountSketch with $b = O(k/\epsilon^2)$ buckets and $t = O(\log (k/\epsilon))$ tables, and compute estimates $\hat{x}_i$ for each $i \in A$. Let $B$ be the set of the top $k$ magnitude estimates $\hat{x}_i$, restricted to $i \in A$. 
\item Return the vector $\hat{x}$ supported on $B$ with corresponding estimates $\hat{x}_i$ for each $i \in B$.   
\end{enumerate}

\thmUBsparserecovery*
\begin{proof}
Since the sketch $S$ is linear, we can assume,  w.l.o.g., that $\|x - x_k \|_2^2 = k$.  Then by Theorem \ref{thm:countsketch}, with probability $1-1/\poly(d)$, the CountSketch  in Step 1 returns estimates $\hat{x}_i'$ satisfying
$$|\hat{x}_i' - x_i| \leq \sqrt{\epsilon}$$
simultaneously for all $i = 1, 2, \ldots, d$. It follows that for every $i$ for which $|x_i| \geq 3 \sqrt{\epsilon}$, we have that $i \in A$. To see this, note that if $|x_i| \geq 3 \sqrt{\epsilon}$, then $|\hat{x}_i'| \geq 2 \sqrt{\epsilon}$. On the other hand, if $|\hat{x}_j'| \geq 2 \sqrt{\epsilon}$ for some $j \in [d]$, then $|x_j| \geq \sqrt{\epsilon}$. The number of $j$ for which $|x_j| \geq \sqrt{\epsilon}$ is at most $k + k/\epsilon$, given that $\|x-x_k\|_2^2 = k$. Also, the number of $i$ for which $|\hat{x}_i| \geq 2 \sqrt{\epsilon}$ is at most $k + k/(4 \epsilon)$. So if $|A| = O(k/\epsilon)$ for a sufficiently large constant in the big-Oh, we have that if $|x_i| \geq 3 \sqrt{\epsilon}$, then $i \in A$. 

In Step 2, and by a union bound over all $i \in A$, we have that with probability $1-1/\poly(k/\epsilon)$, simultaneously for every $i \in A$, $|\hat{x}_i - x_i| \leq \epsilon$. 
Letting $H$ be the set of the top $k$ magnitude coordinates of $x$, we have
$\|x_{H \setminus A} - \hat{x}_{H \setminus A}\|_2^2 \leq k \cdot 9 \cdot \epsilon$, since any $i \in H \setminus A$ 
necessarily satisfies $|x_i| < 3 \sqrt{\epsilon}$ and $\hat{x}_i = 0$. 

We also have $\|x_{H \cap B} - \hat{x}_{H \cap B}\|_2^2 \leq k \cdot \epsilon^2$, since for each $i \in H \cap B$, we have  $|\hat{x}_i - x_i| \leq \epsilon$. 

Finally, consider those $i \in H \cap (A \setminus B)$. For each such $i$, there necessarily exists a $j = j(i) \in A \setminus H$ for which $|\hat{x}_{j(i)}| \geq |\hat{x}_i|$ and since $|\hat{x}_{j(i)} - x_{j(i)}| \leq \epsilon$ and $|\hat{x}_i - x_i| \leq \epsilon$, this implies $|x_{j(i)}| \geq |x_i|- 2\epsilon$. Note also by definition of $H$ that $|x_{j(i)}| \leq |x_i|$. Consequently, the sketch solution $\hat{x}$ pays at most 
$$(x_{j(i)} - \hat{x}_{j(i)})^2 + x_i^2
\leq \epsilon^2 + (x_{j(i)} + 2\epsilon)^2
= O(\epsilon^2 + x_{j(i)}^2 + |x_{j(i)}| \epsilon),$$
on this coordinate, whereas the optimal non-sketched solution
$x_k$ pays $x_{j(i)}^2$. As $|A \setminus B| \leq k$, the total
additional cost the sketched solution pays over the optimal
solution is at most 
\begin{eqnarray}\label{eqn:new}
O(k \epsilon^2 + \epsilon \sum_{i \in H \cap (A \setminus B)} |x_{j(i)}|).
\end{eqnarray}
Finally, note that $\sum_{i \in H \cap (A \setminus B)} |x_{j(i)}|$ is at most the $\ell_1$-norm of the largest $k$ 
coordinates in magnitude not in $H$. Since the $\ell_2$-norm of such
coordinates is at most $\sqrt{k}$, the $\ell_1$-norm of
such coordinates is at most $k$. Combining with (\ref{eqn:new}),
the total additional error the sketched solution pays
is $O(\epsilon k)$. 

It follows that
$$\|x-\hat{x}\|_2^2 \leq 
\|x_{H \setminus A} - \hat{x}_{H \setminus A}\|_2^2 +
\|x_{H \cap B} - \hat{x}_{H \cap B}\|_2^2
+ \|x- x_k\|_2^2 + O(\epsilon k)
\leq 
(1+O(\epsilon))\|x-x_k\|_2^2,$$ 
and (\ref{eqn:prop}) follows by rescaling $\epsilon$ by a constant factor. The total number of measurements and overall failure probability follow by Theorem \ref{thm:countsketch}. 
\end{proof}

\section{Information theoretic basics}

We require the following notions from information theory.

\begin{definition}[Entropy and Mutual Information]
The entropy of a random variable $X$ over some support $S$ is
$$H(X) = \sum_{i\in S} p_i \log_2 \frac{1}{p_i} .$$
Given two random variables $X$ and $Y$, the conditional entropy is
$$H(X|Y) = \sum_y H(X|Y=y) \cdot \mathbb{P}[Y=y]$$
and their joint entropy is
$$H(X,Y) = \sum_{x,y} \mathbb{P}[X=x \bigwedge Y=y] \log_2 \frac{1}{\mathbb{P}[X=x \wedge Y=y]}.$$
The mutual information of two random variables is
$$I(X;Y) = H(X) - H(X|Y) $$
\end{definition}

Finally, we require Fano's inequality.

\begin{lemma}[Fano's Inequality]
\label{lem:fano}
Let $X$ be a random variable chosen from domain $\mathcal{X}$ according to distribution $\mu_X$ and let $Y$ be a random variable chosen from domain $\mathcal{Y}$ according to distribution $\mu_Y$. T
hen for any reconstruction function $g:Y\rightarrow X$ with error $\varepsilon_g$, it holds that 
$$H(X|Y) \leq H(\varepsilon_g) + \varepsilon_g \log(|\mathcal{X}|-1).$$
\end{lemma}

\vfill

\end{document}